\documentclass{article}

\usepackage{amsmath}
\usepackage{amssymb}
\usepackage{amsthm}
\usepackage{caption}

\DeclareSymbolFont{rsfscript}{OMS}{rsfs}{m}{n}
\DeclareSymbolFontAlphabet{\mathrsfs}{rsfscript}

\newcounter{global}
\theoremstyle{definition}
\newtheorem{definition}{Definition}
\theoremstyle{plain}
\newtheorem{theorem}[global]{Theorem}

\newtheorem{lemma}[global]{Lemma}
\newtheorem{corollary}[global]{Corollary}
\newtheoremstyle{note}{}{}{}{}{\itshape}{.}{.5em}{}
\theoremstyle{note}
\newtheorem{remark}{Remark}%
\newtheorem{example}{Example}%
\renewcommand\section{%
  \@startsection {section}{1}{\z@}%
  {-3.5ex \@plus -1ex \@minus -.2ex}%
  {2.3ex \@plus.2ex}%
  {\normalfont\large\bfseries}}
\def\itm#1{{\rm(\textit{\romannumeral#1})}}

\makeatletter

\def\itm#1{{\rm(\textit{\romannumeral#1})}}

\def\SD{\ensuremath{\mathrm{S}}}
\def\itm#1{{\rm(\textit{\romannumeral#1})}}

\def\@leftpar{(}
\def\@rightpar{)}
\def\bigarg{\def\@leftpar{\bigl(}\def\@rightpar{\bigr)}}
\def\mul{\@ifnextchar[{\@withmul}{\@withoutmul}}
\def\@withmul[#1]{\ensuremath{\boldsymbol{f}_{\!#1}\@ifnextchar\bgroup{\@witharg}{\relax}}}
\def\@withoutmul{\ensuremath{\boldsymbol{f}\@ifnextchar\bgroup{\@witharg}{\relax}}}
\def\shf{\@ifnextchar[{\@withshf}{\@withoutshf}}
\def\@withshf[#1]{\ensuremath{\boldsymbol{g}_{#1}\@ifnextchar\bgroup{\@witharg}{\relax}}}
\def\@withoutshf{\ensuremath{\boldsymbol{g}\@ifnextchar\bgroup{\@witharg}{\relax}}}
\def\one{\@ifnextchar[{\@withone}{\@withoutone}}
\def\@withone[#1]{\ensuremath{\boldsymbol{1}_{#1}\@ifnextchar\bgroup{\@witharg}{\relax}}}
\def\@withoutone{\ensuremath{\boldsymbol{1}\@ifnextchar\bgroup{\@witharg}{\relax}}}
\def\C{\@ifnextchar[{\@withC}{\@withoutC}}
\def\@withC[#1]{\ensuremath{\boldsymbol{c}_{#1}\@ifnextchar\bgroup{\@witharg}{\relax}}}
\def\@withoutC{\ensuremath{\boldsymbol{c}\@ifnextchar\bgroup{\@witharg}{\relax}}}
\def\@witharg#1{\@leftpar #1\@rightpar}
\def\D{\@ifnextchar[{\@withD}{\@withoutD}}
\def\@withD[#1]{\ensuremath{\mathcal{C}_{#1}\@ifnextchar\bgroup{\@witharg}{\relax}}}
\def\@withoutD{\ensuremath{\mathcal{C}\@ifnextchar\bgroup{\@witharg}{\relax}}}
\def\@witharg#1{\@leftpar #1\@rightpar}

\makeatother

\begin{document}

\title{Closure structures parameterized by systems of isotone Galois connections}

\date{\normalsize%
  Dept. Computer Science, Palacky University Olomouc}

\author{Vilem Vychodil\footnote{%
    e-mail: \texttt{vychodil@binghamton.edu},
    phone: +420 585 634 705,
    fax: +420 585 411 643}}

\maketitle

\begin{abstract}
  We study properties of classes of closure operators and closure systems
  parameterized by systems of isotone Galois connections. The
  parameterizations express stronger requirements on idempotency and monotony
  conditions of closure operators. The present approach extends previous
  approaches to fuzzy closure operators which appeared in analysis of
  object-attribute data with graded attributes and reasoning with if-then
  rules in graded setting and is also related to analogous results developed
  in linear temporal logic. In the paper, we present foundations of the
  operators and include examples of general problems in data analysis where
  such operators appear.
\end{abstract}

\section{Introduction}
In this paper we deal with closure structures which emerge in data-analytical
applications such as formal concept analysis~\cite{GaWi:FCA} of data with
fuzzy attributes~\cite{BeVy:Fcalh,KlYu:FSFL} and approximate reasoning such
as inference of fuzzy if-then rules from data~\cite{BeVy:ICFCA,BeVy:DASFAA}.
In particular, our paper generalizes and extends observations on fuzzy
closure operators and related structures. Since their inception, fuzzy
closure operators have been the subject of extensive research, the most
influential early papers on the topic
include~\cite{BaKo:Spcicfr,BiGe:CsLs,BiGe:Epco,BoDeCoKe:Ocfptbafrbs,Ch:Gcfs,Ge:Gcrfco,MaGh:Fcs,RoEsGaGo:Icoar}. As
it is usual with graded (fuzzy) generalizations of classic notions, there are
several sound ways to introduce closure operators in fuzzy setting. While
most authors agree on the conditions of extensivity and idempotency, which
take the same form as in the classic setting, the approaches differ in the
treatment of the monotony condition. There are two major approaches:
\begin{enumerate}
\item Using the bivalent notion of inclusion of fuzzy sets, where the
  monotony condition can be written as ``$A \subseteq B$ implies
  $\C{A} \subseteq \C{B}$'' and means that ``the closure of $A$ is
  \emph{fully contained} in the closure of $B$ whenever $A$ is \emph{fully
    contained} in $B$.'' The full containment of fuzzy sets (here
  denoted~``$\subseteq$'') is defined as a bivalent relation on fuzzy sets so
  that for any fuzzy sets $C$ and $D$ in the universe $X$, we put
  $C \subseteq D$ whenever for each element~$x \in X$, the degree to which
  $x$ belongs to $D$ is at least as high as the degree to which it belongs to
  $C$.
\item The second option uses a graded notion of inclusion of fuzzy sets. In
  this case, the monotony condition can be written as
  $\SD(A,B) \leq \SD(\C{A},\C{B})$, where $\leq$ is the order on the set of
  truth degrees (the usual order of reals if the scale of degrees is the real
  unit interval) and $\SD$ is a suitable graded subsethood. Both $\SD(A,B)$
  and $\SD(\C{A},\C{B})$ are general degrees of inclusion, i.e., $\SD(A,B)$
  is the degree to which $A$ is included in $B$ and analogously for
  $\SD(\C{A},\C{B})$. Hence, the monotony condition can be read ``the
  \emph{degree} to which the closure of $A$ is included in the closure of $B$
  is \emph{at least as high} as the \emph{degree} to which $A$ is included in
  $B$.'' In other words, $\SD(A,B)$ gives a lower bound of the inclusion
  degree of the closure of $A$ in the closure of~$B$. In the context of
  approximate inference, this is a desirable property because one can obtain
  a lower approximation of the degree $\SD(\C{A},\C{B})$ without the need to
  actually compute either of $\C{A}$ and $\C{B}$.
\end{enumerate}
These two basic approaches can be seen as two borderline requirements on the
monotony condition and for reasonable choices of $\SD$, which includes the
residuum-based fuzzy set inclusion proposed by Goguen~\cite{Gog:Lic}, the
second approach constitutes a stronger requirement than the first
one. Interestingly, both the approaches can be handled by a single theory
which leaves the approaches as special cases. In fact, there are several
results on fuzzy closure operators where both the approaches result as
special cases. The initial paper~\cite{Be:Fco} by Belohlavek uses a general
monotony condition which is parameterized by an order-filter on the set of
truth degrees. Conceptually different approach has been introduced
in~\cite{BeFuVy:Fcots} where the authors employ linguistic hedges, again, as
parameters of the monotony condition. In both the approaches, the two basic
notions of monotony result by chosen parameterizations---either a special
filter in case of~\cite{Be:Fco} or a special hedge in case
of~\cite{BeFuVy:Fcots}.

In our recent paper~\cite{Vy:Pasefai}, we have developed a theory of graded
if-then rules with general semantics parameterized by systems of isotone
Galois connections. In this setting, general fuzzy closure operators with
parameterized idempotency conditions appeared. Interestingly, in the approach
to attribute implications with temporal semantics introduced first
in~\cite{TrVy:TAsisaits} and developed further in~\cite{TrVy:Ltai}, we have
utilized conceptually similar structures which utilize the notion of being
closed under ``time shifts.'' In this paper, we present results showing that
most results related to closure structures in~\cite{Vy:Pasefai}
and~\cite{TrVy:Ltai} can be handled by a single theory of closure structures
defined on complete lattices which are parameterized by systems of isotone
Galois connections. In addition, we also show that the approaches
in~\cite{Be:Fco,BeFuVy:Fcots} result as special cases of the presented
formalism. Therefore, the present paper studies closure structures from the
perspective of a general class of parameterizations, makes conclusions on a
general level, and particular results like those
in~\cite{Be:Fco,BeFuVy:Fcots,TrVy:Ltai,Vy:Pasefai} can be obtained by
selecting concrete parameterizations on complete lattices.

Our paper is organized as follows. In Section~\ref{sec:prelim}, we present a
survey of notions related to closure operators and closure systems and
introduce notation which is used further in the paper. In
Section~\ref{sec:def}, we present the notions of closure operators and
closure systems parameterized by systems of isotone Galois connections and
show their relationship to parameterized closure structures studied in the
past. In Section~\ref{sec:examples}, we present details on two important
fields related to data analysis where the closure structures parameterized by
systems of isotone Galois connections appear either in the general setting or
as special cases. In Section~\ref{sec:props}, we investigate general
properties of the closure structures and their parameterizations. We give
conclusion and final remarks in Section~\ref{sec:conclusion}.

\section{Preliminaries}\label{sec:prelim}
In the paper, we use the usual notions from the theory of ordered sets
and lattices~\cite{Bir:LT,DaPr}. Recall that a partial order $\leq$ on
an non-empty set $L$ is a binary relation which is reflexive,
antisymmetric, and transitive; the pair $\langle L,\leq\rangle$ is
called a partially ordered set. Furthermore, $\langle L,\leq\rangle$
is called a complete lattice and denoted $\mathbf{L} = \langle
L,\leq\rangle$ whenever each $K \subseteq L$ has its supremum and
infimum in $L$ which are denoted by $\bigvee K$ and $\bigwedge K$,
respectively. Each complete lattice $\mathbf{L}$ has its greatest and
least elements $1 = \bigvee L = \bigwedge \emptyset$ and
$0 = \bigwedge L = \bigvee \emptyset$.

A non-empty subset $K \subseteq L$ is called an $\leq$-filter in
$\mathbf{L}$ if for every $a,b \in L$ such that $a \leq b$ we have
$b \in K$ whenever $a \in K$.

A lattice element $a \in L$ is called compact whenever $a \leq \bigvee J$
for $J \subseteq L$ implies there is a finite $J' \subseteq J$ such that
$a \leq \bigvee J'$. A complete lattice $\mathbf{L}$ is called algebraic (or
compactly generated) whenever each $a \in L$ can be expressed as
$a = \bigvee K$ where $K$ is some subset of $L$ consisting solely of compact
elements.

\subsection{Isotone Galois connections}
Let $\mathbf{L} = \langle L,\leq\rangle$ be a complete lattice. A~pair
$\langle \mul,\shf\rangle$ of operators $\mul\!: L \to L$ and $\shf\!:
L \to L$ is called an \emph{isotone Galois connection} in $\mathbf{L}$
whenever
\begin{align}
  \mul{a} \leq b \text{ if{}f } a \leq \shf{b}
  \label{eqn:gal}
\end{align}
for all $a,b \in L$; $\mul$ is called the \emph{lower adjoint} of~$\shf$ and,
dually, $\shf$ is called the \emph{upper adjoint} of~$\mul$. In an isotone
Galois connection $\langle \mul,\shf\rangle$, $\mul$ uniquely determines
$\shf$ and \emph{vice versa}. In particular,
\begin{align}
  \mul{a} &= \textstyle\bigwedge\{b \in L;\, a \leq \shf{b}\},
  \\
  \shf{b} &= \textstyle\bigvee\{a \in L;\, \mul{a} \leq b\}.
\end{align}
In our paper, we utilize the following properties which are consequences
of~\eqref{eqn:gal}. For any $a,b \in L$ and $a_i \in L$ ($i \in I$), we have:
\begin{align}
  &a \leq \shf{\mul{a}},
  \label{eqn:shfmul} \\
  &\mul{\shf{b}} \leq b
  \label{eqn:mulshf} \\
  &a \leq b \text{ implies } \mul{a} \leq \mul{b},
  \label{eqn:mon_mul} \\
  &a \leq b \text{ implies } \shf{a} \leq \shf{b},
  \label{eqn:mon_shf} \\
  &\bigarg \mul{\textstyle\bigvee\{a_i;\, i \in I\}} =
  \textstyle\bigvee\bigl\{\mul{a_i};\, i \in I\bigr\},
  \label{eqn:mul_distr} \\
  &\bigarg \shf{\textstyle\bigwedge\{a_i;\, i \in I\}} =
  \textstyle\bigwedge\bigl\{\shf{a_i};\, i \in I\bigr\}.
\end{align}
A \emph{composition} of isotone Galois connections is defined in terms
of the ordinary composition of maps. That is, for isotone Galois
connections $\langle\mul[1],\shf[1]\rangle$ and
$\langle\mul[2],\shf[2]\rangle$ in $\mathbf{L}$, we put
\begin{align}
  \langle \mul[1],\shf[1]\rangle \circ \langle\mul[2],\shf[2]\rangle
  = \langle \mul[1]\mul[2],\shf[2]\shf[1]\rangle,
  \label{eqn:compose}
\end{align}
where $\mul[1]\mul[2]$ is a composed operator such that
$\mul[1]\mul[2](a) = \mul[1](\mul[2](a))$ for all $a \in L$ and
analogously for $\shf[2]\shf[1]$. It is easy to see that the
composition is again an isotone Galois connection in $\mathbf{L}$.

We denote by $\one[L]$ the identity operator in $L$, i.e., $\one[L]{a}
= a$ for all $a \in L$. Obviously, $\langle \one[L],\one[L]\rangle$ is
an isotone Galois connection in $\mathbf{L}$. If $\mathbf{L}$ is clear
from the context, we write just $\one$ to denote $\one[L]$. As a
consequence of the fact that $\circ$ is associative and
$\langle\one,\one\rangle$ is neutral with respect to $\circ$, it
follows that all isotone Galois connections in $\mathbf{L}$ together
with binary operation $\circ$ defined by~\eqref{eqn:compose} and
$\langle\one,\one\rangle$ form a monoid.

\subsection{Residuated lattices and related structures}
Residuated lattices~\cite{DiWa} are structures based on lattices which are
enriched by a couple of binary operations satisfying an additional condition.
The structures are widely used in fuzzy and substructural
logics~\cite{CiHaNo1,CiHaNo2,GaJiKoOn:RL} and include structures of degrees
based on left-continuous triangular norms~\cite{KMP:TN} which are popular in
applications~\cite{KlYu:FSFL}. From the point of view of isotonne Galois
connections, such structures can be seen as lattices endowed by particular
systems of isotone Galois connections.

An ordered structure
$\mathbf{L} = \langle L,\leq,\otimes,\rightarrow,0,1\rangle$ is called a
(\emph{commutative integral}) \emph{complete residuated lattice} whenever
$\langle L,\leq\rangle$ is a complete lattice with $0$ and $1$ being the least
and greatest elements, respectively, $\otimes$ is a binary operation in $L$
(called a \emph{multiplication}) which is associative and commutative with $1$
being its neutral element, and $\rightarrow$ is a binary operation in $L$
(called a \emph{residuum}) such that
\begin{align}
  a \otimes b \leq c \text{ if{}f } b \leq a \rightarrow c
  \label{eqn:adj}
\end{align}
holds for all $a,b,c \in L$. In the terminology of residuated lattices,
\eqref{eqn:adj} is called the \emph{adjointness property.} In fuzzy
logics~\cite{EsGoNo:Fotfltc,Haj:MFL}, $\otimes$ and $\rightarrow$ are used as
truth functions of (fuzzy) logical connectives ``conjunction'' and
``implication'', respectively. Alternatively, residuated lattices can be
introduced in terms of isotonne Galois connections as follows. Let
$\mathbf{L} = \langle L,\leq\rangle$ be a complete lattice and let $\otimes$
and $\rightarrow$ be binary operations in $L$ such that $\otimes$ is
associative, commutative, and neutral with respect to $1$. In this setting,
for any $a \in L$, we define maps $\mul[a\otimes]\!: L \to L$ and
$\shf[a\rightarrow]\!: L \to L$ by
\begin{align}
  \mul[a\otimes](b) &= a \otimes b,
  \label{eqn:mul_L}
  \\
  \shf[a\rightarrow](b) &= a \rightarrow b
  \label{eqn:shf_L}
\end{align}
for all $b \in L$. Now,
$\mathbf{L} = \langle L,\leq,\otimes,\rightarrow,0,1\rangle$ is a complete
residuated lattice if and only if
$\langle\mul[a\otimes],\shf[a\rightarrow]\rangle$ is an isotone Galois
connection for any $a \in L$. Indeed, $a \otimes b \leq c$ if{}f
$\mul[a\otimes]{b} \leq c$ if{}f $b \leq \shf[a\rightarrow]{c}$ if{}f
$b \leq a \rightarrow c$ provided that
$\langle\mul[a\otimes],\shf[a\rightarrow]\rangle$ is an isotone Galois
connection and, conversely, $\mul[a\otimes]{b} \leq c$ if{}f
$a \otimes b \leq c$ if{}f $b \leq a \rightarrow c$ if{}f
$b \leq \shf[a\rightarrow]{c}$ provided that~\eqref{eqn:adj} holds.

In the paper, we are going to use general complete as well as concrete
residuated lattices that are used in problem domains related to data analysis
and approximate inference. Namely, we are going to use residuated lattices of
fuzzy sets. Such structures can be understood as direct powers of complete
residuated lattices that serve as structures of truth degrees. In a more
detail, let $\mathbf{L} = \langle L,\leq,\otimes,\rightarrow,0,1\rangle$ be a
complete residuated lattice and let $Y \ne \emptyset$ be a universe
set. Then, the \emph{complete residuated lattice of $\mathbf{L}$-sets in}
(\emph{the universe}) $Y$ is a structure
$\mathbf{L}^Y = \langle
L^Y,\subseteq,\otimes_Y,\rightarrow_Y,0_Y,1_Y\rangle$, where
\begin{itemize}\parskip=0pt%
\item
  $L^Y$ is the set of all maps of the form $A\!: Y \to L$,
  each $A \in L^Y$ is called an $\mathbf{L}$-fuzzy set
  (shortly, an $\mathbf{L}$-set) in $Y$, see~\cite{Gog:LFS,Gog:Lic};
\item
  $\subseteq$ is a binary relation on $L^Y$ such that $A \subseteq B$
  whenever $A(y) \leq B(y)$ for all $y \in Y$;
\item
  $\otimes_Y$ and $\rightarrow_Y$ are defined componentwise using $\otimes$
  and $\rightarrow$, i.e., for any $A,B \in L^Y$ and $y \in Y$, we have
  \begin{align}
    (A \otimes_Y B)(y) &= A(y) \otimes B(y),
    \\
    (A \rightarrow_Y B)(y) &= A(y) \rightarrow B(y);
  \end{align}
  If there is no danger of confusion and both $\mathbf{L}$ and $Y$ are
  clear from context, we write just $\otimes$ and $\rightarrow$ to denote
  $\otimes_Y$ and $\rightarrow_Y$;
\item
  $0_Y \in L^Y$ and $1_Y \in L^Y$ so that $0_Y(y) = 0$
  and $1_Y(y) = 1$ for all $y \in Y$.
\end{itemize}
It follows from the basic properties of complete residuated lattices that
$\mathbf{L}^Y$ is a complete residuated lattice (the class of complete
lattices is closed under arbitrary direct
products~\cite{Bir:LT,Wec:UAfCS}). In examples, we are going to use the usual
notation for writing $\mathbf{L}$-sets, e.g., $\{y^a,z^b\}$ represents
$A\!: \{y,z\} \to L$ such that $A(y) = a$ and $A(z) = b$.

The relation $\subseteq$ in $\mathbf{L}^Y$ is called
a \emph{subsethood} (or \emph{inclusion relation}) of $\mathbf{L}$-sets and
$A \subseteq B$ expresses the fact that $A$ is fully contained in $B$.
If $A(y)$ and $B(y)$ are interpreted as degrees to which $y \in Y$ belongs
to $A$ and $B$ respectively, it follows that $A \subseteq B$ means that,
for each $y \in Y$, the degree to which $y$ belongs to $B$ is at least as
high as the degree to which $y$ belongs to $A$. In addition to this type
of ``full subsethood'', it is reasonable to define its graded counterpart
which expresses general degrees~\cite{Gog:LFS,Gog:Lic}
to which one $\mathbf{L}$-set is included
in another one. For $A,B \in L^Y$, we put
\begin{align}
  \SD(A,B) &= \textstyle\bigwedge\{A(y) \rightarrow B(y);\, y \in Y\}
  \label{eqn:S}
\end{align}
and call $\SD(A,B) \in L$ the \emph{subsethood degree} (of $A$ in $B$). That
is, $\SD$ defined by~\eqref{eqn:S} is a map of the form
$\SD\!: L^Y \times L^Y \to L$. It can be easily seen that $A \subseteq B$
if{}f $A(y) \rightarrow B(y) = 1$ for all $y \in Y$ which is if{}f
$\SD(A,B) = 1$.

In addition to our understanding of $\otimes$ and $\rightarrow$ as operations
on $\mathbf{L}$ and $\mathbf{L}^Y$ (which are, in fact,
$\otimes_Y$ and $\rightarrow_Y$), we also consider $\otimes$ and $\rightarrow$
as maps $\otimes\!: L \times L^Y \to L^Y$ and
$\rightarrow: L \times L^Y \to L^Y$ given by
\begin{align}
  (a \otimes A)(y) &= a \otimes A(y),
  \label{eqn:a_multiple}
  \\
  (a \rightarrow A)(y) &= a \rightarrow A(y),
  \label{eqn:a_shift}
\end{align}
for all $A \in L^Y$ and $a \in L$. For fixed $a \in L$, $a \otimes A$,
and $a \rightarrow A$ given by~\eqref{eqn:a_multiple}
and~\eqref{eqn:a_shift} are called the \emph{$a$-multiple} and
\emph{$a$-shift} of $A$, respectively. Note that $\rightarrow$ is not
commutative, i.e., one may also consider $A \rightarrow a$ but this
operation is not relevant to our investigation.

Using~\eqref{eqn:adj}, \eqref{eqn:S}, \eqref{eqn:a_multiple}, and
\eqref{eqn:a_shift}, we derive the following property which is extensively
used in our paper:
\begin{align}
  a \otimes B \subseteq C
  \text{ if{}f }
  B \subseteq a \rightarrow C
  \text{ if{}f }
  a \leq \SD(B,C)
  \label{eqn:alt_adj}
\end{align}
for any $a \in L$ and any $B,C \in L^Y$.
Indeed, $a \otimes B \subseteq C$ if{}f $a \otimes B(y) \leq C(y)$
for any $y \in Y$, i.e., by~\eqref{eqn:adj},
we get $B(y) \leq a \rightarrow C(y)$ for any $y \in Y$, meaning
$B \subseteq a \rightarrow C$. Furthermore, using the commutativity
of $\otimes$ and~\eqref{eqn:adj} twice, it follows that
$B(y) \leq a \rightarrow C(y)$ for any $y \in Y$ if{}f
$a \leq B(y) \rightarrow C(y)$ for any $y \in Y$ which is if{}f
$a \leq \SD(B,C)$ because $\SD(B,C)$ is the greatest lower bound of
all $B(y) \rightarrow C(y)$ where $y \in Y$.

\subsection{Closure structures}\label{sec:clos_old}
In this section, we recall two influential types of closure operators defined
in complete residuated lattices of $\mathbf{L}$-sets. The operators differ
in definitions of the isotony condition which is in both cases stronger
than the ordinary isotony.

The approach in~\cite{BeFuVy:Fcots} introduced $\mathbf{L}^*$-closure
operators as operators on $L^Y$ whose isotony condition is
parameterized by truth-stressing linguistic
hedges~\cite{Za:FL,Za:Afstilh}. In our setting a \emph{truth-stressing
  linguistic hedge} (shortly, a hedge) on $\mathbf{L}$ is any map
${}^*\!: L \to L$ such that
\begin{align}
  1^* &= 1, \label{ts:1}
  \\
  a^* &\leq a, \label{ts:sub}
  \\
  (a \rightarrow b)^* &\leq a^* \rightarrow b^*, \label{ts:mon*}
\end{align}
for all $a,b \in L$.

\begin{remark}\label{rem:hajek}
  The conditions~\eqref{ts:1}--\eqref{ts:mon*} are a subset of conditions of
  truth-stressing hedges as they were studied by H\'ajek in~\cite{Haj:Ovt}
  and interpreted as truth functions of logical connectives ``very true.'' As
  it is argued in~\cite{Haj:Ovt}, \eqref{ts:1}--\eqref{ts:mon*} may be
  considered natural properties of (truth functions of) logical connectives
  ``very true'': \eqref{ts:1} says that ``$1$ (degree denoting the full
  truth) is very true''; \eqref{ts:sub} reflects the fact that if a
  proposition is considered ``very true'', then it is also considered
  ``true'' (i.e., being ``very true'' is at least as strong as being
  ``true'') and~\eqref{ts:mon*} says that from a very true implication with a
  very true antecedent, one derives a very true consequent. Notice
  that~\eqref{ts:1} and~\eqref{ts:mon*} yield
  \begin{align}
    \text{if } a \leq b \text{ then } a^* \leq b^*
    \label{ts:mon}
  \end{align}
  for all $a,b \in L$. In fact, in this paper, we rely on \eqref{ts:mon},
  i.e., all the considerations can be made for hedges
  satisfying~\eqref{ts:1}, \eqref{ts:sub}, and~\eqref{ts:mon}.
\end{remark}

Most of the applications in relational data analysis as well as other
results where closure structures parameterized by hedges appear rely
on \emph{idempotent truth-stressing hedges}, i.e., ${}^*$ is in addition
required to satisfy 
\begin{align}
  a^* &= a^{**} \label{ts:idm}
\end{align}
for all $a \in L$.

\begin{remark}\label{rem:idm*}
  The idempotency, as a property of hedges, is disputable. In case of the
  truth-stressing hedges, it is widely accepted that ``very very true'' is
  strictly stronger an emphasis than ``very
  true''~\cite{Za:FL,Za:Afstilh,Za:lv1,Za:lv2,Za:lv3}. Indeed, $x^* = x^2$ is
  often taken as a truth function for a hedge if $\mathbf{L}$ is defined on
  the real unit interval using a left-continuous triangular norm acting as
  $\otimes$ which is a typical choice in applications. In contrast,
  \cite{HaHa:AhGfl} argues that in case of truth-depressing (or
  truth-intensifying) hedges~\cite{Vy:Tdhbl}, the idempotency may seem
  natural. For instance, it is not so frequent that ``more or less'' is
  chained in order to further lessen the impact of an
  utterance. Nevertheless, major results in relational data
  analysis~\cite{BeVy:Fcalh} and logics of if-then
  rules~\cite{BeVy:Addg1,BeVy:Addg2,BeVy:FHLI,BeVy:FHLII} where the hedges
  were employed use~\eqref{ts:idm} and the condition cannot be dropped
  without losing important properties of the studied concepts. We refer
  readers interested in treatment of hedges in fuzzy logics to recent
  paper~\cite{EsGoNo:Hedges} and~\cite[Section~VII--3]{CiHaNo2}.
\end{remark}

Fuzzy closure operators in $\mathbf{L}^Y$ parameterized by hedges have been
introduced in~\cite{BeFuVy:Fcots} as follows:

\begin{definition}\label{def:Cts}
  Let $\mathbf{L}$ be a complete residuated lattice, ${}^*$ be a hedge
  satisfying~\eqref{ts:1}--\eqref{ts:mon*}, $Y$ be a non-empty universe set.
  An operator $\C\!: L^Y \to L^Y$ is called
  an $\mathbf{L}^*$-closure operator~\cite{BeFuVy:Fcots} in $\mathbf{L}^Y$
  whenever
  \begin{align}
    A &\subseteq \C{A},
    \label{cl:ext} \\
    \SD(A,B)^* &\leq \SD(\C{A},\C{B}),
    \label{cl:*mon} \\
    \C{\C{A}} &\subseteq \C{A},
    \label{cl:idm}
  \end{align}
  for all $A,B \in L^Y$.
\end{definition}

Recall that $\SD$ in~\eqref{cl:*mon} is the residuum-based graded
subsethood defined by~\eqref{eqn:S}. Therefore, taking into account
the interpretation of ${}^*$ as a truth function of connective ``very
true'', \eqref{cl:*mon} can be read: ``If it is very true that $A$ is
included in $B$, then $\C{A}$ is included in $\C{B}$''. A finer
reading, which involves explicit references to degrees is: ``The
degree to which $\C{A}$ is included in $\C{B}$ is at least as high as
the degree to which it is very true that $A$ is included in $B$''.

Let us stress that $\mathbf{L}^*$-closure operators are indeed
ordinary closure operators, i.e., \eqref{cl:*mon} implies the ordinary
isotony condition. Indeed, if $A \subseteq B$, then $\SD(A,B) = 1$,
i.e., applying~\eqref{ts:1}, we get that
\begin{align}
  \text{if } A \subseteq B \text{ then } \C{A} \subseteq \C{B}.
  \label{cl:mon}
\end{align}
In general, \eqref{cl:*mon} is stronger than~\eqref{cl:mon} at it is
shown in~\cite[Remark 2.3]{BeFuVy:Fcots}. The conditions become equivalent if
${}^*$ is the so-called \emph{globalization}~\cite{TaTi:Gist}, i.e.,
if for any $a \in L$, we have
\begin{align}
  a^* &=
  \begin{cases}
    1, &\text{if } a = 1, \\
    0, &\text{otherwise.}
  \end{cases}
  \label{eqn:glob}
\end{align}
If ${}^*$ is globalization, then $A \nsubseteq B$ gives $\SD(A,B)^* =
0$ because $\SD(A,B) < 1$. Therefore, \eqref{cl:*mon} is indeed
equivalent to \eqref{cl:mon}.

The second important approach~\cite{Be:Fco} to parameterized closure
operators on complete residuated lattices of $\mathbf{L}$-sets, which
predates the approach by hedges, is based on the notion of an $\leq$-filter
which has been recalled in the beginning of Section~\ref{sec:prelim}.

\begin{definition}\label{def:LK}
  Let $\mathbf{L}$ be a complete residuated lattice, $K \subseteq L$ be an
  $\leq$-filter in~$\mathbf{L}$, and $Y$ be a non-empty universe set. An
  operator $\C\!: L^Y \to L^Y$ is called an $\mathbf{L}_K$-closure
  operator~\cite{Be:Fco} in $\mathbf{L}^Y$ whenever \eqref{cl:ext},
  \eqref{cl:idm}, and the following condition:
  \begin{align}
    \text{if } \SD(A,B) \in K \text{ then } \SD(A,B) \leq \SD(\C{A},\C{B})
    \label{cl:Kmon}
  \end{align}
  are satisfied for all $A,B \in L^Y$.
\end{definition}

Analogously as in the case of $\mathbf{L}^*$-closure operators, we can see
that condition~\eqref{cl:Kmon} implies~\eqref{cl:mon} because $1 \in K$ and
thus $A \subseteq B$ yields $\SD(A,B) = 1 \in K$ which gives
$\SD(\C{A},\C{B}) = 1$, meaning $\C{A} \subseteq \C{B}$. In addition,
\eqref{cl:Kmon} becomes~\eqref{cl:mon} for $K = \{1\}$ which is trivially an
$\leq$-filter in $\mathbf{L}$. Hence, following the notation in
Definition~\ref{def:LK}, we refer to the ordinary closure operators in
$\mathbf{L}^Y$ as to the $\mathbf{L}_{\{1\}}$-closure operators,
see~\cite{Be:Fco}.

\section{$\mathbf{S}$-closure operators}\label{sec:def}%
In this section, we introduce $S$-closure operators and related
notions. Furthermore, we focus on their relationship to the established
closure operators parameterized by hedges and $\leq$-filters which have been
recalled in Section~\ref{sec:prelim}. We show that both the parameterizations
may be seen as special cases of the parameterization by systems of isotone
Galois connections. In addition, the studied operators can be seen as
generalizations of those appearing in~\cite{TrVy:TAsisaits,Vy:Pasefai}.

\begin{definition}\label{def:S-clop}
  Let $\mathbf{L} = \langle L,\leq\rangle$ be a complete lattice. Any set $S$
  of isotone Galois connections in $\mathbf{L}$ such that
  $\langle \one,\one\rangle \in S$ is called an $L$-parameterization.  An
  operator $\C\!: \langle L,\leq\rangle \to \langle L,\leq\rangle$ is called
  an $S$-closure operator in $\langle L, \leq\rangle$ whenever
  \begin{align}
    a &\leq \C{a},
    \label{S:ext} \\
    a \leq b &\text{ implies } \C{a} \leq \C{b},
    \label{S:mon} \\
    \C{\shf{\C{a}}} &\leq \shf{\C{a}},
    \label{S:idm}
  \end{align}
  are satisfied for all $a,b \in L$ and all $\langle \mul,\shf\rangle \in S$.
  If $S$ is closed under compositions, i.e., if
  $\mathbf{S} = \langle S,\circ,\langle\one,\one\rangle\rangle$ is a monoid,
  we call it an $\mathbf{L}$-parameterization and, in addition, $\C$
  satisfying \eqref{S:ext}--\eqref{S:idm} is called an $\mathbf{S}$-closure
  operator.
\end{definition}

It is evident that for $S = \{\langle\one,\one\rangle\}$, all
$S$-closure operators in $\mathbf{L}$ are exactly the classic closure
operators in $\mathbf{L}$. Indeed, for $\langle\mul,\shf\rangle$ being
$\langle\one,\one\rangle$, \eqref{S:idm} becomes
\begin{align}
  \C{\C{a}} &\leq \C{a}
  \label{S:ord_idm}
\end{align}
which together with~\eqref{S:ext} yields the classic idempotency
condition. It is easy to see that without the requirement of
$\langle\one,\one\rangle \in S$, the operator may not be idempotent in
general.

We present examples of $\mathbf{S}$-closure operators which appear in several
fields of relational data analysis in Section~\ref{sec:examples}. In the rest
of this section, we investigate the relationship of $S$-closure operators to
the operators in residuated lattices of $\mathbf{L}$-sets summarized in
Section~\ref{sec:clos_old}.

\begin{theorem}\label{th:L*_S}
  Let $\mathbf{L}$ be a complete residuated lattice, $Y \ne \emptyset$, and
  let ${}^*$ be an idempotent truth-stressing hedge. An operator
  $\C\!: L^Y \to L^Y$ is an $\mathbf{L}^*$-closure operator if{}f it is an
  $\mathbf{L}_{\{1\}}$-closure operator and
  \begin{align}
    \C{a^* \rightarrow \C{A}} &\subseteq a^* \rightarrow \C{A}
    \label{cl:*shf}
  \end{align}
  holds for all $a \in L$ and $A \in L^Y$.
\end{theorem}
\begin{proof}
  Let $\C$ be an $\mathbf{L}^*$-closure operator. Obvously,
  \eqref{cl:mon} is a particular case of~\eqref{cl:*mon},
  see~\cite[Remark 2.3]{BeFuVy:Fcots}.
 As a consequence,
  $\C$ is an $\mathbf{L}_{\{1\}}$-closure operator.
  Therefore, it remains to show
  that $\C$ satisfies~\eqref{cl:*shf}. We prove this fact using~\eqref{ts:idm},
  \eqref{cl:*mon}, and \eqref{cl:idm}. Take any $a \in L$ and $A \in L^Y$.
  Using~\eqref{eqn:alt_adj}, from
  \begin{align*}
    a^* \rightarrow \C{A} &\subseteq a^* \rightarrow \C{A}
  \end{align*}
  it follows that 
  \begin{align*}
    a^* \otimes (a^* \rightarrow \C{A}) &\subseteq \C{A}
  \end{align*}
  which, again by~\eqref{eqn:alt_adj}, gives
  \begin{align*}
    a^* \leq \SD(a^* \rightarrow \C{A}, \C{A}).
  \end{align*}
  Now, using \eqref{ts:mon}, \eqref{ts:idm}, \eqref{cl:*mon},
  and \eqref{cl:idm}, the last inequality yields
  \begin{align*}
    a^* &= a^{**}
    \\
    &\leq \SD(a^* \rightarrow \C{A}, \C{A})^*
    \\
    &\leq \SD(\C{a^* \rightarrow \C{A}}, \C{\C{A}})
    \\
    &\leq \SD(\C{a^* \rightarrow \C{A}}, \C{A}),
  \end{align*}
  i.e., $a^* \leq \SD(\C{a^* \rightarrow \C{A}}, \C{A})$.
  Now, using~\eqref{eqn:alt_adj} twice, we get
  \begin{align*}
    a^* \otimes \C{a^* \rightarrow \C{A}} \subseteq \C{A}
  \end{align*}
  and finally
  \begin{align*}
    \C{a^* \rightarrow \C{A}} \subseteq a^* \rightarrow \C{A},
  \end{align*}
  i.e., \eqref{cl:*shf} is satisfied.

  Conversely, suppose that $\C$ is an $\mathbf{L}_{\{1\}}$-closure operator
  such that \eqref{cl:*shf} holds. We show that $\C$ satisfies~\eqref{cl:*mon}.
  Take any $A,B \in L^Y$.
  Observe that using~\eqref{cl:*shf} for $a = \SD(A,B)$, we get
  \begin{align*}
    \C{\SD(A,B)^* \rightarrow \C{B}} \subseteq \SD(A,B)^* \rightarrow \C{B}
  \end{align*}
  from which, owing to~\eqref{eqn:alt_adj}, it follows that
  \begin{align*}
    \SD(A,B)^* \leq \SD(\C{\SD(A,B)^* \rightarrow \C{B}},\C{B}).
  \end{align*}
  Since the graded subsethood is antitone in the first argument,
  we obtain \eqref{cl:*mon} as a consequence of the previous inequality
  and the fact that $\C{A} \subseteq \C{\SD(A,B)^* \rightarrow \C{B}}$
  which is indeed true: Using~\eqref{ts:sub}, \eqref{cl:ext},
  and the isotony of $\SD$
  in the second argument, we get
  \begin{align*}
    \SD(A,B)^* \leq \SD(A,B) \leq \SD(A,\C{B})
  \end{align*}
  and thus $A \subseteq \SD(A,B)^* \rightarrow \C{B}$ by~\eqref{eqn:alt_adj}.
  Using~\eqref{cl:mon}, we further get
  \begin{align*}
    \C{A} \subseteq \C{\SD(A,B)^* \rightarrow \C{B}}
  \end{align*}
  As a consequence,
  \begin{align*}
    \SD(A,B)^* \leq
    \SD(\C{\SD(A,B)^* \rightarrow \C{B}},\C{B}) \leq \SD(\C{A},\C{B})
  \end{align*}
  which proves~\eqref{cl:*mon}.
\end{proof}

\begin{figure}
  \centering%
  \begin{tabular}{|r|cccc|}
    \hline
    $\otimes$ & $0$ & $a$ & $b$ & $1$ \\
    \hline
    $0$ & $0$ & $0$ & $0$ & $0$ \\
    $a$ & $0$ & $0$ & $0$ & $a$ \\
    $b$ & $0$ & $0$ & $0$ & $b$ \\
    $1$ & $0$ & $a$ & $b$ & $1$ \\
    \hline
  \end{tabular}
  \qquad
  \begin{tabular}{|r|cccc|}
    \hline
    $\rightarrow$ & $0$ & $a$ & $b$ & $1$ \\
    \hline
    $0$ & $1$ & $1$ & $1$ & $1$ \\
    $a$ & $b$ & $1$ & $1$ & $1$ \\
    $b$ & $b$ & $b$ & $1$ & $1$ \\
    $1$ & $0$ & $a$ & $b$ & $1$ \\
    \hline
  \end{tabular}
  \caption{Operations $\otimes$ and $\rightarrow$
    used in Remark~\ref{rem:L*_S}\,(a).}
  \label{fig:ts_cntr1}
\end{figure}

\begin{figure}
  \centering%
  \begin{tabular}{|r|cccc|}
    \hline
    $\otimes$ & $0$ & $a$ & $b$ & $1$ \\
    \hline
    $0$ & $0$ & $0$ & $0$ & $0$ \\
    $a$ & $0$ & $0$ & $0$ & $a$ \\
    $b$ & $0$ & $0$ & $b$ & $b$ \\
    $1$ & $0$ & $a$ & $b$ & $1$ \\
    \hline
  \end{tabular}
  \qquad
  \begin{tabular}{|r|cccc|}
    \hline
    $\rightarrow$ & $0$ & $a$ & $b$ & $1$ \\
    \hline
    $0$ & $1$ & $1$ & $1$ & $1$ \\
    $a$ & $b$ & $1$ & $1$ & $1$ \\
    $b$ & $a$ & $a$ & $1$ & $1$ \\
    $1$ & $0$ & $a$ & $b$ & $1$ \\
    \hline
  \end{tabular}
  \caption{Operations $\otimes$ and $\rightarrow$
    used in Remark~\ref{rem:L*_S}\,(b).}
  \label{fig:ts_cntr2}
\end{figure}

\begin{remark}\label{rem:L*_S}
  (a) Observe that in the only-if part of the proof of Theorem~\ref{th:L*_S},
  we have utilized the isotony~\eqref{ts:mon} of ${}^*$ instead of the
  stronger condition~\eqref{ts:mon*}. Also, in the only-if part of the proof,
  we have not utilized the subdiagonality
  condition~\eqref{ts:sub}. Neither~\eqref{ts:1} nor~\eqref{ts:mon}
  nor~\eqref{ts:idm} can be omitted because otherwise the only-if part of the
  assertion would not hold. This is obvious in the case of~\eqref{ts:1}. In
  order to see that~\eqref{ts:mon} is necessary, suppose that $\mathbf{L}$ is
  defined on a four-element linearly ordered set $L = \{0, a, b, 1\}$ so that
  $0 < a < b < 1$. It can be easily checked that $\otimes$ and $\rightarrow$
  defined by the tables in Figure~\ref{fig:ts_cntr1} satisfy the adjointness
  property. Moreover, we can consider ${}^*\!: L \to L$ such that $c^* = c$
  for $c \in \{0,a,1\}$ and $b^* = 0$. Obviously, ${}^*$
  satisfies~\eqref{ts:1}, \eqref{ts:idm}, and it does not
  satisfy~\eqref{ts:mon}. Now, take $Y = \{y\}$ and consider an operator
  $\C\!: L^{\{y\}} \to L^{\{y\}}$ such that
  \begin{align}
    \C{A} &=
    \begin{cases}
      \{y^0\}, &\text{for } A = \{y^0\}, \\
      \{y^1\}, &\text{otherwise.}
    \end{cases}
    \label{eqn:ex_C_top}
  \end{align}
  By a routine check, it follows that $\C$ satisfies~\eqref{cl:ext},
  \eqref{cl:*mon}, and \eqref{cl:idm}. On the other hand, \eqref{cl:*shf} is
  not satisfied. Indeed, for $A = \{y^0\}$, we have
  \begin{align*}
    \C{a^* \rightarrow \C{\{y^0\}}} &= 
    \C{a \rightarrow \{y^0\}} \\
    &=
    \C{\{y^{a \rightarrow 0}\}} \\
    &=
    \C{\{y^b\}} \\
    &=
    \{y^1\} \\
    &\nsubseteq
    \{y^b\}
    =
    a^* \rightarrow \C{\{y^0\}}.
  \end{align*}
  Therefore, at least~\eqref{ts:mon} is necessary in order to establish the
  assertion of Theorem~\ref{th:L*_S}.

  (b) Analogously as in (a), we may proceed for an operator ${}^*$ which is
  not idempotent. Consider $\mathbf{L}$ with $\otimes$ and $\rightarrow$
  defined as in Figure~\ref{fig:ts_cntr2}. Furthermore, consider ${}^*$ such
  that $0^* = a^* = 0$, $b^* = a$, and $1^* = 1$. Considering the same order
  of elements in $L$ as before, $\mathbf{L}$ is a complete residuated lattice
  and ${}^*$ satisfies~\eqref{ts:1}, \eqref{ts:sub},
  and~\eqref{ts:mon*}. Thus, ${}^*$ is a truth-stressing hedge which is not
  idempotent because $b^{**} \ne b^{*}$. Furthermore,
  $\C\!: L^{\{y\}} \to L^{\{y\}}$ defined by
  \begin{align*}
    \C{A} &=
    \begin{cases}
      A, &\text{if } A = \{y^0\} \text{ or } A = \{y^a\}, \\
      \{y^1\}, &\text{otherwise.}
    \end{cases}
  \end{align*}
  is an $\mathbf{L}^*$-closure operator. It is easy to see that
  \begin{align*}
    \C{b^* \rightarrow \C{\{y^0\}}}
    &=
    \C{a \rightarrow \C{\{y^0\}}} \\
    &=
    \C{a \rightarrow \{y^0\}} \\
    &=
    \C{\{y^b\}} \\
    &=
    \{y^1\} \\
    &\nsubseteq
    \{y^b\} = 
    b^* \rightarrow \C{\{y^0\}},
  \end{align*}
  i.e., \eqref{cl:*shf} does not hold.

  (c) In the if-part of Theorem~\ref{th:L*_S}, the subdiagonality
  condition~\eqref{ts:sub} is also necessary. If $\mathbf{L}$ is the
  two-valued Boolean algebra with $L = \{0,1\}$ and $0 < 1$, then for
  a map ${}^*\!: L \to L$ such that $0^* = 1^* = 1$ it follows that
  the identity operator $\one[L^{\!\{y\}}]$ satisfies~\eqref{cl:*shf}. On the
  contrary, we have
  \begin{align*}
    \SD(\{y^1\},\{y^0\})^* = 0^* = 1 \nleq 0 =
    \SD(\{y^1\},\{y^0\}) =
    \SD(\one[L^{\!\{y\}}]{\{y^1\}},\one[L^{\!\{y\}}]{\{y^0\}}),
  \end{align*}
  i.e., $\one[L^{\!\{y\}}]$ violates~\eqref{cl:*mon}.
\end{remark}

As a consequence of Theorem~\ref{th:L*_S}, we get a corollary presented below
which states that $\mathbf{L}^*$-closure operators are particular
$\mathbf{S}$-closure operators provided that ${}^*$ is idempotent. In the
corollary, we use lower and upper adjoints written as $\mul[a \otimes]$ and
$\shf[a \rightarrow]$ for some $a \in L$ which are in fact defined
analogously as~\eqref{eqn:a_multiple} and \eqref{eqn:a_shift}, respectively.
That is, by putting
\begin{align}
  \mul[a \otimes]{B} &= a \otimes B, \\
  \shf[a \rightarrow]{B} &= a \rightarrow B,
\end{align}
for all $a \in L$ and $B \in L^Y$ and considering~\eqref{eqn:alt_adj}
together with Theorem~\ref{th:L*_S}, we have the following observation:

\begin{corollary}\label{cor:L*_S}
  Let $\mathbf{L}$ be a complete residuated lattice, $Y \ne \emptyset$, and
  let ${}^*$ be an idempotent truth-stressing hedge.  An operator
  $\C\!: L^Y \to L^Y$ is an $\mathbf{L}^*$-closure operator if{}f it is an
  $\mathbf{S}$-closure operator for
  $S = \{\langle \mul[a^* \otimes], \shf[a^* \rightarrow]\rangle;\, a \in
  L\}$.  \qed
\end{corollary}

Let us note that in Corollary~\ref{cor:L*_S}, the considered operator is
indeed an $\mathbf{S}$-closure operator and not just $S$-closure operator,
i.e., $S$ is closed under the composition of isotone Galois connections.
Indeed, for any $a,b \in L$ and $c = a^* \otimes b^*$, we have
$\mul[a^* \otimes]\mul[b^* \otimes] = \mul[c^* \otimes ]$ and
$\shf[b^* \rightarrow]\shf[a^* \rightarrow] = \shf[c^* \rightarrow ]$ which
follows from the fact that $a^* \otimes b^* = (a^* \otimes b^*)^*$ provided
that the hedge ${}^*$ is idempotent, see~\cite[Lemma 2]{BeVy:MRAP},
cf. also~\cite[Example 1\,(b)]{Vy:Pasefai}

We now turn our attention to the relationship to $\mathbf{L}_K$-closure
operators. Analogously as in Theorem~\ref{th:L*_S}, we may establish the
following characterization.

\begin{theorem}\label{th:LK_S}
  An operator $\C\!: L^Y \to L^Y$ is an $\mathbf{L}_K$-closure
  operator if{}f it is an $\mathbf{L}_{\{1\}}$-closure operator and
  \begin{align}
    \C{a \rightarrow \C{A}} &\subseteq a \rightarrow \C{A}
    \label{cl:Kshf}
  \end{align}
  holds for all $a \in K$ and $A \in L^Y$.
\end{theorem}
\begin{proof}
  Most parts of the proof use similar arguments as in the proof of
  Theorem~\ref{th:L*_S}. Therefore, we only comment on the technical
  differences.  Suppose that $\C$ is an $\mathbf{L}_K$-closure
  operator. Clearly, $\C$ satisfies~\eqref{cl:mon} because $1 \in
  K$. Take any $a \in K$ and $A \in L^Y$. By~\eqref{eqn:alt_adj},
  we get
  \begin{align*}
    a \leq \SD(a \rightarrow \C{A}, \C{A})
  \end{align*}
  and thus $\SD(a \rightarrow \C{A}, \C{A}) \in K$ because $K$ is
  an $\leq$-filter. Therefore, applying~\eqref{cl:Kmon}, \eqref{cl:idm},
  and the isotony of $\SD$ in the second argument, it follows that
  \begin{align*}
    a 
    &\leq
    \SD(a \rightarrow \C{A}, \C{A}) \\
    &\leq
    \SD(\C{a \rightarrow \C{A}}, \C{\C{A}}) \\
    &\leq
    \SD(\C{a \rightarrow \C{A}}, \C{A}).
  \end{align*}
  Thus, using~\eqref{eqn:alt_adj}, we obtain~\eqref{cl:Kshf}.

  Conversely, assuming that $\C$ is an $\mathbf{L}_{\{1\}}$-closure operator
  satisfying~\eqref{cl:Kshf}, we prove that \eqref{cl:Kmon} holds.  Take
  $A,B \in L^Y$ such that $\SD(A,B) \in K$. As a particular case
  of~\eqref{cl:Kshf}, we get
  \begin{align*}
    \C{\SD(A,B) \rightarrow \C{B}} &\subseteq \SD(A,B) \rightarrow \C{B}
  \end{align*}
  and thus, using~\eqref{eqn:alt_adj}, we get
  \begin{align*}
    \SD(A,B) \leq \SD(\C{\SD(A,B) \rightarrow \C{B}},\C{B}).
  \end{align*}
  In order to finish the proof, we show that
  $\C{A} \subseteq \C{\SD(A,B) \rightarrow \C{B}}$ which results by
  applying~\eqref{cl:ext}, \eqref{cl:mon} together with the isotony of $\SD$
  in the second argument, and~\eqref{eqn:alt_adj} in much the same way as in
  the proof of Theorem~\ref{th:L*_S}.
\end{proof}

\begin{remark}
  Analogously as in the case of Theorem~\ref{th:L*_S}, we can show that the
  condition of $K$ being an $\leq$-filter in Theorem~\ref{th:LK_S} is
  essential. Indeed, take $\mathbf{L}$ with $L = \{0,a,b,1\}$ as in
  Remark~\ref{rem:L*_S} and let $\otimes$ and $\rightarrow$ be defined as in
  Figure~\ref{fig:ts_cntr1}. Furthermore, take $K = \{a,1\}$. Obviously, $K$
  is not an $\leq$-filter because $b \not\in K$. Now, an operator
  $\C: L^{\{y\}} \to L^{\{y\}}$ defined by \eqref{eqn:ex_C_top}
  satisfies~\eqref{cl:ext}, \eqref{cl:idm}, and~\eqref{cl:Kmon}. On the other
  hand, it does not satisfy~\eqref{cl:Kshf}. Indeed, observe that for
  $a \in K$ and $A = \{y^0\}$, we have
  \begin{align*}
    \C{a \rightarrow \C{\{y^0\}}}
    &=
    \C{a \rightarrow \{y^0\}} \\
    &=
    \C{\{y^b\}} \\
    &=
    \{y^1\} \\ 
    &\nsubseteq
    \{y^b\} = a \rightarrow \C{\{y^0\}},
  \end{align*}
  i.e., if $K$ is not an $\leq$-filter, then $\C$ does not
  satisfy~\eqref{cl:Kshf} in general. Note that in the if-part of
  Theorem~\ref{th:LK_S}, no assumptions on $K$ are needed.
\end{remark}

\begin{corollary}\label{cor:LK_S}
  Let $\mathbf{L}$ be a complete residuated lattice, $Y \ne \emptyset$, and
  let $K$ be an $\leq$-filter. An operator $\C\!: L^Y \to L^Y$ is an
  $\mathbf{L}_K$-closure operator if{}f it is an $S$-closure operator for
  $S = \{\langle \mul[a \otimes], \shf[a \rightarrow]\rangle;\, a \in K\}$.
  \qed
\end{corollary}

Unlike~\ref{cor:L*_S}, $S$ in Corollary~\ref{cor:LK_S} is not closed under
compositions in general. As a consequence, there are $\mathbf{L}_K$-operators
which are not $\mathbf{S}$-closure operators for
$S = \{\langle \mul[a \otimes], \shf[a \rightarrow]\rangle;\, a \in K\}$.  For
instance, if $\mathbf{L}$ is the standard \L ukasiewicz algebra, i.e., if $L$
is the real unit interval with its natural ordering and $\otimes$ is given by
$a \otimes b = \max(a+b-1, 0)$, then $K = [0.5,1]$ is obviously an
$\leq$-filter but
$S = \{\langle \mul[a \otimes], \shf[a \rightarrow]\rangle;\, a \geq 0.5\}$ is
not closed under $\circ$ because, e.g.,
$\langle \mul[0.5 \otimes], \shf[0.5 \rightarrow]\rangle \circ \langle
\mul[0.5 \otimes], \shf[0.5 \rightarrow]\rangle = \langle \mul[0 \otimes],
\shf[0 \rightarrow]\rangle \not\in K$.

\section{Examples}\label{sec:examples}
In the following subsections, we show two important areas of reasoning
with if-then rules and extracting information from relational data where
$S$-closure operators naturally appear.

\subsection{Enriched Armstrong-style inference systems}
In this subsection, we show an important example of $\mathbf{S}$-closure
operators induced by enriched Armstrong-style inference
systems~\cite{Arm:Dsdbr} defined on algebraic lattices. We introduce general
inference systems and consider operators which map each element of a lattice
to its syntactic closure which is defined by the inference system. We prove
the operator is indeed an $\mathbf{S}$-closure operator and show that several
inference systems that appeared in the past in different contexts have
syntactic closures which may be viewed as special cases of the general one.

Recall the notions of compactness of lattice elements and (complete)
algebraic lattices from Section~\ref{sec:prelim} and assume that
$\mathbf{L} = \langle L,\leq\rangle$ is a complete algebraic lattice with
$K \subseteq L$ being the set of all its compact elements. We use this
assumption throughout the entire section and we are not going to repeat it.
At this point, we may view $\mathbf{L}$ as an abstract system of elements
which can be used to form particular formulas. Namely, any ordered pair
$\langle a,b\rangle \in K \times K$ is called a (well-formed) \emph{formula}
and for better readability we are going to denote it $a \Rightarrow b$ and
read it ``if $a$ then $b$.'' The intended meaning of $a \Rightarrow b$ is to
express that if a lattice element is at least as high as~$a$, then it is also
at least as high as~$b$. We introduce an inference system for the rules which
resembles the famous Armstrong system for reasoning with functional
dependencies~\cite{Arm:Dsdbr,Mai:TRD}. In our setting, the system is enriched
by new inference rules defined by lower adjoints coming from particular
$L$-parameterizations.

\begin{definition}\label{def:infsyst}
  Let $S$ be an $L$-parameterization. If for any $a \in K$ and any
  $\langle\mul,\shf\rangle \in S$, we have $\mul{a} \in K$, then $S$
  is called a compact $L$-parameterization. Under this assumption, an
  $S$-inference system for $\mathbf{L}$ is a set of the following
  inference rules:
  \begin{enumerate}\parskip=0pt%
  \item
    (from no assumptions) infer $a \vee b \Rightarrow b$,
  \item
    from $a \Rightarrow b$ and $b \vee c \Rightarrow d$ infer
    $a \vee c \Rightarrow d$,
  \item
    from $a \Rightarrow b$ infer $\mul{a} \Rightarrow \mul{b}$
  \end{enumerate}
  for any $a,b,c,d \in K$ and any $\langle\mul,\shf\rangle \in S$.  Let
  $\Sigma$ be a set of formulas. An $S$-proof of $a \Rightarrow b$ by
  $\Sigma$ is any finite sequence $\delta_1,\ldots,\delta_n$ of formulas
  such that $\delta_n$ is $a \Rightarrow b$ and for each $i = 1,\ldots,n$,
  we have that
  \begin{itemize}\parskip=0pt%
  \item
    $\delta_i \in \Sigma$, or
  \item
    $\delta_i$ results from some of the formulas $\delta_j$ ($j < i$) using
    a single application of one of the inference rules 1.--3.
  \end{itemize}
  We put $\Sigma \vdash a \Rightarrow b$ and say that $a \Rightarrow b$ is
  $S$-provable by $\Sigma$ whenever there is an $S$-proof of
  $a \Rightarrow b$ by $\Sigma$.
\end{definition}

\begin{remark}\label{rem:dedprop}
  (a) Note that the inference rules of an $S$-inference system produce only
  well-formed formulas. Indeed, this is a consequence of the following
  facts. First, if $a,b \in K$, then $a \vee b \in K$, i.e., the supremum of
  compact elements in $\mathbf{L}$ is a compact element in $\mathbf{L}$.
  Thus, the inference rules 1. and 2. produce well-formed formulas. Second,
  if $a \in K$, then $\mul{a} \in K$ because we assume that $S$ is a compact
  $L$-parameterization. This shows that 3. always produces a well-formed
  formula. Also note that 1. is in fact an axiom schema saying that each
  formula of the form $a \Rightarrow b$ where $b \leq a$ is $S$-provable by
  any set of formulas.

  (b) The classic system of Armstrong rules can be seen as an inference
  system which is a particular case of that in Definition~\ref{def:infsyst}:
  Take a finite $R$ which is the set of attributes of a relation scheme and
  consider the finite complete lattice
  $\mathbf{L} = \langle2^R,\subseteq\rangle$ of all subsets of $R$. Each
  element of $\mathbf{L}$ is compact because $2^R$ is finite, i.e., $K = L$.
  Then, formulas in our sense are expressions of the form $A \Rightarrow B$,
  where $A,B \in 2^R$ which agrees with the type of formulas used in the
  Armstrong system. Furthermore, the inference rules 1. and 2. become
  \begin{itemize}\parskip=0pt%
  \item
    infer $A \cup B \Rightarrow B$, and
  \item
    from $A \Rightarrow B$ and $B \cup C \Rightarrow D$
    infer $A \cup C \Rightarrow D$,
  \end{itemize}
  which are equivalent to the classic Armstrong rules and are called
  the axiom and pseudo-transitivity in~\cite{Mai:TRD}. Finally,
  for $S = \{\langle\one,\one\rangle\}$, which is trivially a compact
  $L$-parameterization, the last inference rule becomes trivial: from
  $A \Rightarrow B$ infer $A \Rightarrow B$ and can be disregarded.

  (c) As in the case of the classic Armstrong system, we can easily
  derive the following basic properties of $S$-provability~\cite{Mai:TRD}:
  \begin{itemize}\parskip=0pt%
  \item 
    weakening: if $\Sigma \vdash a \Rightarrow c$, then
    $\Sigma \vdash a \vee b \Rightarrow c$,
  \item 
    transitivity: if $\Sigma \vdash a \Rightarrow b$ and
    $\Sigma \vdash b \Rightarrow c$, then
    $\Sigma \vdash a \Rightarrow c$,
  \item addition: if $\Sigma \vdash a \Rightarrow b$ and
    $\Sigma \vdash a \Rightarrow c$, then
    $\Sigma \vdash a \Rightarrow b \vee c$,
  \end{itemize}
  for any $\Sigma$ and any $a,b,c \in K$.
\end{remark}

We now focus on particular $S$-closure operators which are induced by
$S$-inference systems and given sets of formulas and show their basic
properties.

\begin{definition}
  Let $S$ be a compact $L$-parameterization of $\mathbf{L}$.
  For any set $\Sigma$ of formulas, we define operators 
  $\D[\Sigma]\!: L \to 2^K$ and $\C[\Sigma]\!: L \to L$ by
  \begin{align}
    \D[\Sigma]{a} &=
    \{b \in K;\, \Sigma \vdash c \Rightarrow b
    \text{ for some } c \in K \text{ such that } c \leq a\},
    \label{cl:DED}
    \\
    \C[\Sigma]{a} &= \textstyle\bigvee \D[\Sigma]{a}
    \label{cl:ded}
  \end{align}
  for all $a \in L$; $\C[\Sigma]{a}$ is called the syntactic $S$-closure
  of $a \in L$ (using $\Sigma$).
\end{definition}

The following technical observation shows that~\eqref{cl:ded}
can be simplified provided that $a$ is a compact element of $\mathbf{L}$.

\begin{lemma}\label{le:C_com}
  If $a \in K$, then
  $\C[\Sigma]{a} = \bigvee\{b \in K;\, \Sigma \vdash a \Rightarrow b\}$.
\end{lemma}
\begin{proof}
  Let $a \in K$ and let $b \in \D[\Sigma]{a}$. That is, there is
  $c \in K$ such that $c \leq a$ and $\Sigma \vdash c \Rightarrow
  b$. Then, using the principle of weakening, see
  Remark~\ref{rem:dedprop}\,(c), we get
  $\Sigma \vdash a \Rightarrow b$ and thus
  $b \leq \bigvee\{b \in K;\, \Sigma \vdash a \Rightarrow b\}$. Since
  $b \in \D[\Sigma]{a}$ was taken arbitrarily, we get
  $\C[\Sigma]{a} \leq \bigvee\{b \in K;\, \Sigma \vdash a \Rightarrow
  b\}$.

  Conversely, take any $b \in K$ such that $\Sigma \vdash a \Rightarrow b$.
  Since $a \leq a$ and $a \in K$, it follows that $b \in \D[\Sigma]{a}$ from
  which the claim readily follows.
\end{proof}

The operator defined by~\eqref{cl:ded} can be used to check whether a
given $a \Rightarrow b$ ($a,b \in K$) is provable by a given $\Sigma$.
Indeed, the next assertion shows that the problem of deciding
$\Sigma \vdash a \Rightarrow b$ is equivalent to the problem of
determining whether $\C[\Sigma]{a}$ is at least as high as $b$ in terms
of the underlying lattice order.

\begin{theorem}\label{th:dedchar}
  Let $a,b \in K$. We have $\Sigma \vdash a \Rightarrow b$ if{}f\/
  $b \leq \C[\Sigma]{a}$.
\end{theorem}
\begin{proof}
  The only-if part follows directly by Lemma~\ref{le:C_com}.

  Conversely, suppose that $b \leq \C[\Sigma]{a}$ holds.
  Using the fact that both $a$ and $b$ are compact and
  taking into account Lemma~\ref{le:C_com}, we get that there is a finite
  $B = \{b_1,\ldots,b_n\} \subseteq K$ such that $b \leq \bigvee B \in K$
  and $\Sigma \vdash a \Rightarrow b_i$ for all $i=1,\ldots,n$.
  Now, $b \leq \bigvee B \in K$ yields
  $\Sigma \vdash b_1 \vee \cdots \vee b_n \Rightarrow b$. Furthermore,
  using the addition principle multiple times, see
  Remark~\ref{rem:dedprop}\,(c), it follows that
  $\Sigma \vdash a \Rightarrow b_1 \vee \cdots \vee b_n$. Thus, using the
  transitivity of $\vdash$, we get $\Sigma \vdash a \Rightarrow b$.
\end{proof}

Note that similar characterizations of entailment as that in
Theorem~\ref{th:dedchar} appear in various approaches which involve if-then
rules. For instance, in logic programming~\cite{Lloyd84}, an analogous role
is played by least models of definite programs. The next assertion shows that
$\C[\Sigma]$ is indeed an $S$-closure operator which is in addition
\emph{algebraic} in the sense that the closure of any element $a \in L$ can
be obtained as the supremum of closures of all compact elements which are
smaller than or equal to $a$. In addition, we show that $\C[\Sigma]$ is an
$\mathbf{S}'$-closure operator where $\mathbf{S}'$ is the uniquely given
$\mathbf{L}$-parameterization generated by $S$. That is, $S'$ is the least
subset of isotone Galois connections in $L$ such that $S \subseteq S'$ and
$S'$ is closed under $\circ$. In such a case, we say that
$\mathbf{S'} = \langle S',\circ,\langle\one,\one\rangle\rangle$ is a
\emph{monoid} (\emph{of isotone Galois connections}) \emph{generated by $S$.}

\begin{theorem}\label{th:cSigma}
  Let $\Sigma$ be a theory and $S$ be a compact $L$-parameterization. Then,
  $\C[\Sigma]$ defined by~\eqref{cl:ded} is an algebraic $\mathbf{S'}$-closure
  operator where
  $\mathbf{S'} = \langle S',\circ,\langle\one,\one\rangle\rangle$ is
  the monoid generated by $S$.
\end{theorem}
\begin{proof}
  First, we prove that~$\C[\Sigma]$ given by~\eqref{cl:ded} is extensive,
  i.e., it satisfies~\eqref{S:ext}. Since $\mathbf{L}$ is algebraic, any
  $a \in L$ can be expressed as $a = \bigvee B$ where $B \subseteq K$.  Take
  any $b \in B$. It suffices to check that $b \leq \C[\Sigma]{a}$.  Observe
  that the facts that $b \in K$, $b \leq a$, and
  $\Sigma \vdash b \Rightarrow b$ yield that $b \in \D[\Sigma]{a}$ and thus
  $b \leq \C[\Sigma]{a}$. Altogether, $\C[\Sigma]$ is extensive.

  Second, we check that $\C[\Sigma]$ satisfies~\eqref{S:mon}. This is easy to
  see. Observe that assuming $a_1 \leq a_2$, the fact $b \in \D[\Sigma]{a_1}$
  means there is $c \in K$ such that $c \leq a_1$ and
  $\Sigma \vdash c \Rightarrow b$. But in that case, $c \leq a_1 \leq a_2$
  and so $b \in \D[\Sigma]{a_2}$ from which we immediately get
  $\C[\Sigma]{a_1} \leq \C[\Sigma]{a_2}$.
  
  Third, we prove that $\C[\Sigma]$ satisfies~\eqref{S:idm}. Consider any
  $\langle\mul,\shf\rangle$ which results as a composition of arbitrary
  finitely many isotone Galois connections in $S$. We check that
  $\C[\Sigma]{\shf{\C[\Sigma]{a}}} \leq \shf{\C[\Sigma]{a}}$ for any
  $a \in L$. Take $b \in \D[\Sigma]{\shf{\C[\Sigma]{a}}}$. It suffices to
  check that $b \leq \shf{\C[\Sigma]{a}}$. Using~\eqref{cl:DED}, there is
  $c \in K$ such that $\Sigma \vdash c \Rightarrow b$ and
  $c \leq \shf{\C[\Sigma]{a}}$. Using~\eqref{eqn:gal}, it follows that
  $\mul{c} \leq \C[\Sigma]{a}$ and, in addition, $\mul{c}$ is a compact
  element because $S$ is a compact $L$-parameterization. Therefore, from
  $\mul{c} \leq \C[\Sigma]{a}$, it follows that there are
  $\{b_1,\ldots,b_n\} \subseteq \D[\Sigma]{a}$ and
  $\{c_1,\ldots,c_n\} \subseteq K$ such that
  $\mul{c} \leq b_1 \vee \cdots \vee b_n \in K$, $c_i \leq a$, and
  $\Sigma \vdash c_i \Rightarrow b_i$ for all $i=1,\ldots,n$. Hence,
  $\Sigma \vdash b_1 \vee \cdots \vee b_n \Rightarrow \mul{c}$ and
  $\Sigma \vdash c_1 \vee \cdots \vee c_n \Rightarrow b_1 \vee \cdots \vee
  b_n$ from which by the transitivity of $\vdash$ it follows that
  $\Sigma \vdash c_1 \vee \cdots \vee c_n \Rightarrow \mul{c}$. Moreover, we
  have observed that $\Sigma \vdash c \Rightarrow b$, i.e., applying the
  third inference rule multiple times (recall that $\mul$ is in fact a
  composition of finitely many lower adjoints from $S$), we get
  $\Sigma \vdash \mul{c} \Rightarrow \mul{b}$. Therefore, applying the
  transitivity of $\vdash$ once again, we get
  $\Sigma \vdash c_1 \vee \cdots \vee c_n \Rightarrow \mul{b}$. At this
  point, we have $c_1 \vee \cdots \vee c_n \in K$ such that
  $c_1 \vee \cdots \vee c_n \leq a$ and
  $\Sigma \vdash c_1 \vee \cdots \vee c_n \Rightarrow \mul{b}$ for
  $\mul{b} \in K$, i.e., using~\eqref{cl:DED}, we get
  $\mul{b} \in \D[\Sigma]{a}$. Thus, $\mul{b} \leq \C[\Sigma]{a}$ and so
  $b \leq \shf{\C[\Sigma]{a}}$.

  Finally, it remains to prove that $\C[\Sigma]$ is algebraic. Take any
  $a \in L$ and $b \in \D[\Sigma]{a} \subseteq K$, i.e., there is $c \in K$
  such that $c \leq a$ and $\Sigma \vdash c \Rightarrow b$. Hence,
  $b \in \D[\Sigma]{c}$, i.e., $b \leq \C[\Sigma]{c}$. As a consequence,
  $\C[\Sigma]{a} \leq \bigvee\{\C[\Sigma]{c};\, c \in K \text{ and } c \leq
  a\}$.
\end{proof}

In addition to the previous characterization which associates algebraic
$S$-closure operators to given theories, we show a converse characterization:
Each algebraic $S$-closure operator over a compact $L$-parameterization can
be seen as an operator associated to some theory.  This observation is
summarized in the next theorem based on the following lemma.

\begin{lemma}\label{le:syncl}
  Let $\C$ be an algebraic $S$-closure operator where $S$ is a compact
  $L$-parameterization and let
  \begin{align}
    \Sigma_{\C} &= \{a \Rightarrow b;\, a,b \in K \text{ and\/ } b \leq \C{a}\}.
    \label{eqn:SigmaC}
  \end{align}
  Then, for any $a,b \in K$, if\/ $\Sigma_{\C} \vdash a \Rightarrow b$, then
  $a \Rightarrow b \in \Sigma_{\C}$.
\end{lemma}
\begin{proof}
  First, observe that if $a \Rightarrow b$ is an instance of the axiom
  schema, i.e., if{}f $b \leq a$, then~\eqref{S:ext} gives
  $b \leq a \leq \C{a}$ and thus $a \Rightarrow b \in \Sigma_{\C}$.  Hence,
  if a formula results using the first inference rule in
  Definition~\ref{def:infsyst}, then it belongs to $\Sigma_{\C}$.

  Take any $a \Rightarrow b \in \Sigma_{\C}$ and
  $b \vee c \Rightarrow d \in \Sigma_{\C}$. We show that the formula
  $a \vee c \Rightarrow d$, which results by the application of the second
  inference rule in Definition~\ref{def:infsyst}, belongs to $\Sigma_{\C}$ as
  well. Using~\eqref{eqn:SigmaC}, we get that $b \leq \C{a}$ and
  $d \leq \C{b \vee c}$. Therefore, using~\eqref{S:mon}
  and~\eqref{S:ord_idm}, it follows that
  \begin{align*}
    d \leq
    \C{b \vee c} \leq
    \C{\C{a} \vee c} \leq
    \C{\C{a \vee c}} \leq
    \C{a \vee c}.
  \end{align*}
  Hence, $a \vee c \Rightarrow d \in \Sigma_{\C}$.

  Take any $\langle\mul,\shf\rangle \in S$ and any
  $a \Rightarrow b \in \Sigma_{\C}$. We show that
  $\mul{a} \Rightarrow \mul{b}$, which results by the application of the
  third inference rule in Definition~\ref{def:infsyst}, belongs
  to~$\Sigma_{\C}$. Using~\eqref{eqn:SigmaC}, the assumption
  $a \Rightarrow b \in \Sigma_{\C}$ yields $b \leq \C{a}$ and thus, owing
  to~\eqref{eqn:mon_mul}, we get $\mul{b} \leq \mul{\C{a}}$.  Thus, it
  suffices to check that $\mul{\C{a}} \leq \C{\mul{a}}$ which is indeed the
  case. We postpone the proof of this general and important inequality to the
  next section, see Theorem~\ref{th:mulC_Cmul}.

  Now, the claim of Lemma~\ref{le:syncl} follows directly by our observations
  using induction over the length of an $S$-proof by $\Sigma_{\C}$.
\end{proof}

\begin{theorem}\label{th:Sigmac}
  Let $\C$ be an algebraic $S$-closure operator where $S$ is a compact
  $L$-parameterization. Then, for $\Sigma_{\C}$ given by~\eqref{eqn:SigmaC},
  we have $\C = \C[\Sigma_{\C}]$.
\end{theorem}
\begin{proof}
  Owing to Theorem~\ref{th:cSigma}, $\C[\Sigma_{\C}]$ is an algebraic
  $S$-closure operator and thus it suffices to show that
  $\C{a} = \C[\Sigma_{\C}]{a}$ for each $a \in K$.

  First, take $a \in K$ and let $b \in K$ such that $b \leq \C{a}$.
  Using~\eqref{eqn:SigmaC}, we get $a \Rightarrow b \in \Sigma_{\C}$ and thus
  $\Sigma_{\C} \vdash a \Rightarrow b$. Therefore, using
  Theorem~\ref{th:dedchar}, we have $b \leq \C[\Sigma_{\C}](a)$.  Since
  $\mathbf{L}$ is an algebraic lattice and $b \in K$ was taken arbitrarily,
  we get that $\C{a} \leq \C[\Sigma_{\C}]{a}$ for $a \in K$.

  Conversely, for any $a \in K$ and any $b \in K$, we show that
  $b \leq \C{a}$ provided that $b \leq \C[\Sigma_{\C}]{a}$. Using
  Theorem~\ref{th:dedchar} and the assumption $b \leq \C[\Sigma_{\C}]{a}$, it
  follows that $\Sigma_{\C} \vdash a \Rightarrow b$. Using
  Lemma~\ref{le:syncl}, the last fact yields
  $a \Rightarrow b \in \Sigma_{\C}$. Now, directly from~\eqref{eqn:SigmaC},
  we get $b \leq \C{a}$ which finishes the proof.
\end{proof}

\begin{remark}
  Let us note that the presented inference system for if-then formulas which
  is parameterized by a system of isotone Galois connections is not the only
  one possible. Several equivalent systems can be introduced which are
  analogous to the graph-based system in~\cite{UrVy:Dddosbd} and the system
  based on the simplification equivalence~\cite{BeCoEnMoVy:Apaddg} which is
  suitable for automated provers. Also note that Theorem~\ref{th:cSigma} and
  Theorem~\ref{th:Sigmac} are limited to algebraic $S$-closure operators on
  algebraic lattices. The assertions can be extended to unrestricted
  $S$-closure operators on arbitrary complete lattices by extending the
  inference system by an infinitary deduction rule in a similar way as in the
  approaches in \cite{BeVy:Falcrl,KuVy:Flprrai}.
\end{remark}

Finally, we show the relationship of the general approach presented in this
section to the well-established approaches to reasoning with various types of
if-then dependencies.

\begin{example}\label{ex:syst}
  (a) As we have mentioned in Remark~\ref{rem:dedprop}, the classic Armstrong
  system for reasoning with functional dependencies is a particular case of
  the presented system. In addition, owing to the equivalence of the semantic
  entailment of functional dependencies and particular formulas in the
  classic propositional logic, see~\cite{DeCa,Fagin,SaDePaFa:Ebrddfpl}, the
  inference system can also be used for reasoning with attribute
  implications~\cite{GuDu} which appear as basic formulas in formal concept
  analysis~\cite{GaWi:FCA}.

  (b) Our approach generalizes the inference systems for fuzzy attribute
  implications~\cite{BeVy:ICFCA} and similarity-based functional
  dependencies~\cite{BeVy:DASFAA} parameterized by linguistic hedges, see
  also~\cite{BeVy:Addg1,BeVy:Addg2} for detailed survey. In these approaches,
  $\mathbf{L}$ is the residuated lattice of all $\mathbf{L}$-sets in a finite
  universe $Y$. In this setting, formulas are considered as expressions
  $A \Rightarrow B$ where $A,B \in L^Y$ and are called fuzzy (graded)
  attribute implications. The inference system for fuzzy attribute
  implications consists of inference rules which correspond to 1. and 2. from
  Definition~\ref{def:infsyst}. In addition, the inference system uses the
  following rule of multiplication: From $A \Rightarrow B$, infer
  $c^* {\otimes} A \Rightarrow c^* {\otimes} B$. Therefore, it can be easily
  seen that the inference system corresponds to the general system presented
  in this section for
  $S = \{\langle\mul[c^* \otimes],\shf[c^* \rightarrow]\rangle;\, c \in L\}$
  with $\mul[c^* \otimes]$ and $\shf[c^* \rightarrow]$ given as in
  \eqref{eqn:mul_L} and \eqref{eqn:shf_L}, respectively. The approach to
  parameterization by linguistic hedges was further developed
  in~\cite{Vy:Pasefai} where we have presented fuzzy if-then rules with
  semantics parameterized by systems of isotone Galois connections;
  \cite{Vy:Pasefai} is in fact the major motivation for the present paper
  which looks at the closure structures from a more general perspective
  because it is not concerned only with complete residuated lattices of fuzzy
  sets and extends the observations on the closure structures studied
  initially in~\cite{Vy:Pasefai}.

  (c) The present approach also generalizes an inference system for
  attribute implications with temporal semantics. Attribute implications
  annotated by time points have been introduced in~\cite{TrVy:TAsisaits}
  and investigated in more detail in~\cite{TrVy:Ltai}.
  From the point of view of the present formalism, the formulas are
  expressions of the form $A \Rightarrow B$ such that
  $A,B \subseteq Y \times \mathbb{Z}$, where $Y$ is a finite set of
  attributes, both $A$ and $B$ are finite, and the integers from $\mathbb{Z}$
  are used to denote relative time points. For instance,
  $A = \{\langle y,-1\rangle, \langle z,2\rangle\} \subseteq Y \times
  \mathbb{Z}$ is interpreted as expressing the fact that $y$ holds in the
  time point prior to the present one (i.e., $-1$ time units from now) and
  $z$ holds in the time point after the next one (i.e., $+2$ time units from
  now). A complete Armstrong-style inference system for such formulas is
  presented in~\cite{TrVy:Ltai}. It can be seen as a particular case of
  the general system presented in this section by considering $\mathbf{L}$
  with $L = 2^{Y \times \mathbb{Z}}$ and $\subseteq$ being the complete
  lattice order on $L$. Obviously, all finite
  $A \subseteq Y \times \mathbb{Z}$ are exactly the compact elements of
  $\mathbf{L}$ and, in addition, $\mathbf{L}$ is an algebraic
  lattice. Furthermore, put
  $\boldsymbol{s}(A) = \{\langle y,i+1\rangle;\, \langle y,i\rangle \in A\}$
  for any $A \subseteq Y \times \mathbb{Z}$ and consider its inverse
  $\boldsymbol{s}^{-1}$. In this setting, the inference system
  from~\cite{TrVy:Ltai} becomes the general inference system for
  $S = \{\langle\one,\one\rangle,
  \langle\boldsymbol{s},\boldsymbol{s}^{-1}\rangle\}$.
\end{example}

\subsection{Consequence operators induced by data}
One of the basic problems in FCA~\cite{GaWi:FCA,Wi:RLT} is discovery of
attribute dependencies described by sets of attribute implications which hold
in given data. In this subsection, we approach the problem from the
perspective of closure operators and general parameterizations. Analogously
as in the previous subsection, we show that the notion of an if-then
dependency being true in data can be defined on a general level and the basic
criterion for a dependency being true in data can be expressed using fixed
points of the parameterized closure operators studied in our paper.

Recall that in the classic FCA~\cite{GaWi:FCA,GuDu}, the input data is an
object-attribute data table, called a (dyadic) formal context, which is
formalized as a binary relation $I \subseteq X \times Y$ between the set $X$
of objects and the set $Y$ of attributes under consideration. An attribute
implication $A \Rightarrow B$ where $A,B \subseteq Y$ is considered true in
$I$ whenever $A \subseteq \{x\}^{\uparrow_I}$ implies
$B \subseteq \{x\}^{\uparrow_I}$ for all $x \in X$ where ${}^{\uparrow_I}$ is
the standard concept-forming operator, i.e.,
$\{x\}^{\uparrow_I} = \{y \in Y;\, \langle x,y\rangle \in I\}$.  It is
well-known that $A \Rightarrow B$ being true in $I$ can equivalently be
expressed by means of the inclusion of $B$ in the intent of $I$ generated by
$A$. We know show an analogous result on a more general level: First,
instead of a dyadic formal context, we take a subset $M \subseteq L$ where
$\mathbf{L} = \langle L,\leq\rangle$ is a complete lattice; $M$ can be seen
as representing ``input data.'' Attribute implications will be considered in
much the same way as in the previous subsection, and their interpretation in
$M$ is defined using an $L$-parameterization $S$:

\begin{definition}\label{def:Msat}
  Let $M \subseteq L$ and let $S$ be an $L$-parameterization.
  For any $a,b \in L$, we put $M \models a \Rightarrow b$
  and say that $a \Rightarrow b$ is true in $M$ whenever
  the following condition is satisfied:
  \begin{align}
    \text{if } \mul{a} \leq m \text{ then } \mul{b} \leq m
  \end{align}
  for all $m \in M$ and any $\langle\mul,\shf\rangle$ such that
  \begin{align}
    \langle\mul,\shf\rangle &=
    \langle\mul[1],\shf[1]\rangle \circ \cdots \circ
    \langle\mul[n],\shf[n]\rangle 
    \label{eqn:mulcirc}
  \end{align}
  for some 
  $\langle\mul[1],\shf[1]\rangle,\ldots,\langle\mul[n],\shf[n]\rangle \in S$.
\end{definition}

\begin{remark}\label{rem:Msat}
  (a) The classic notion of an attribute implication being true in a formal
  context is generalized by Definition~\ref{def:Msat} as follows:
  $\mathbf{L}$ is the complete lattice of all subsets of $Y$ ordered by the
  usual set inclusion and for each formal context $I \subseteq X \times Y$
  we take $M = \{\{x\}^{\uparrow_I};\, x \in X\} \subseteq L =
  2^Y$. In this setting, for $S = \{\langle \one,\one\rangle\}$,
  $\models$ from Definition~\ref{def:Msat} becomes the usual notion of
  satisfaction of an attribute implication in a formal context.

  (b) It is easy to see that $\models$ from Definition~\ref{def:Msat} can be
  seen as a relation of satisfaction (i.e., $M \models a \Rightarrow b$ may
  be understood so that $a \Rightarrow b$ is satisfied in~$M$) which
  generalizes properties of the classic relation of satisfaction of attribute
  implications in formal contexts. For instance, $M \models a \Rightarrow c$
  implies $M \models a \vee b \Rightarrow c$ which can be seen as a form of
  the law of weakening; $M \models a \Rightarrow b \vee c$ implies
  $M \models a \Rightarrow c$; $M \models a \Rightarrow b$ and
  $M \models b \Rightarrow c$ imply $M \models a \Rightarrow c$ which can be
  seen as a form of the transitivity of $\models$.
\end{remark}

With each $M \subseteq L$ and $a \in L$ which serves as an antecedent of a
formula we associate a set of all consequents $b \in L$ such that
$M \models a \Rightarrow b$. Sets of this form are introduced in the
following definition.

\begin{definition}
  Let $S$ be an $L$-parameterization. For any $M \subseteq L$ we define
  operators $\D[M]\!: L \to 2^L$ and $\C[M]\!: L \to L$ by
  \begin{align}
    \D[M]{a} &= \{b \in L;\, M \models a \Rightarrow b\}, \label{eqn:DM}
    \\
    \C[M]{a} &= \textstyle\bigvee\D[M]{a},  \label{eqn:CM}
  \end{align}
  for all $a \in L$; $\C[M]{a}$ is called the semantic $S$-closure of $a \in L$
  (using $M$).
\end{definition}

The following assertion shows that $\D[M]{a}$ is closed under arbitrary
suprema and thus $\C[M]{a}$ belongs to $\D[M]{a}$.

\begin{lemma}\label{le:supclose}
  Let $M \subseteq L$ and let $S$ be an $L$-parameterization.
  For any $a \in L$, we have $\C[M]{a} \in \D[M]{a}$.
\end{lemma}
\begin{proof}
  Take any $\{b_i \in L;\, i \in I\} \subseteq \D[M]{a}$. We prove that the
  supremum of $\{b_i \in L;\, i \in I\}$ belongs to $\D[M]{a}$ and thus we
  obtain the proof of Lemma~\ref{le:supclose} as a particular case of this
  observation. Using~\eqref{eqn:DM}, $b_i \in \D[M]{a}$ ($i \in I$) means that
  $M \models a \Rightarrow b_i$ ($i \in I$).
  Thus, if $\mul{a} \leq m$ for some $m \in M$ and $\langle\mul,\shf\rangle$
  of the form~\eqref{eqn:mulcirc}, then $\mul{b_i} \leq m$ for all $i \in I$.
  Therefore, $\bigvee\{\mul{b_i};\, i \in I\} \leq m$. Now, we can
  apply~\eqref{eqn:mul_distr} to get 
  $\mul{\bigvee\{b_i;\, i \in I\}} \leq m$.
  Hence, $M \models a \Rightarrow \bigvee\{b_i;\, i \in I\}$ and 
  so $\bigvee\{b_i;\, i \in I\} \in \D[M]{a}$.
\end{proof}

As an immediate consequence of Lemma~\ref{le:supclose}, we get that
$\C[M]{a}$ is the greatest element of $\D[M]{a}$. Hence, using~\eqref{eqn:DM}
together with Lemma~\ref{le:supclose}, we get the following corollary which
generalizes the standard criterion for classic attribute implications being
true in dyadic formal contexts, cf. Remark~\ref{rem:Msat}\,(b).

\begin{corollary}\label{co:Mchar}
  For any $M \subseteq L$ and $a,b \in L$,
  $M \models a \Rightarrow b$ if{}f\/ $b \leq \C[M]{a}$.
  \qed
\end{corollary}

In the following observations, we show that $\C[M]$ defined by~\eqref{eqn:CM}
is an $S$-closure operator. Actually, we prove more than that. The operator
is an $\mathbf{S}'$-closure operator for $\mathbf{S}'$ being the
$\mathbf{L}$-parameterization generated by $S$.

\begin{lemma}\label{le:Mmul}
  Let $M \subseteq L$, $S$ be an $L$-parameterization, and $a,b \in L$ such
  that $M \models a \Rightarrow b$. Then,
  $M \models \mul{a} \Rightarrow \mul{b}$ for all $\langle\mul,\shf\rangle$
  of the form~\eqref{eqn:mulcirc}.
\end{lemma}
\begin{proof}
  Take any $\langle\mul[1],\shf[1]\rangle,\langle\mul[2],\shf[2]\rangle$
  of the form~\eqref{eqn:mulcirc} and suppose that
  $\mul[1]{\mul[2]{a}} \leq m$ for some $m \in M$.
  Since $M \models a \Rightarrow b$ and $\mul[1] \circ \mul[2]$ is also
  of the form~\eqref{eqn:mulcirc}, it readily follows that 
  $\mul[1]{\mul[2]{b}} \leq m$. Hence,
  $M \models \mul[2]{a} \Rightarrow \mul[2]{b}$ as a consequence of
  the fact that $\langle\mul[1],\shf[1]\rangle$ was taken arbitrarily.
\end{proof}

\begin{theorem}
  Let $M \subseteq L$ and let $S$ be an $L$-parameterization.
  Then, $\C[M]$ defined by~\eqref{eqn:CM} is an $\mathbf{S'}$-closure
  operator where
  $\mathbf{S'} = \langle S',\circ,\langle\one,\one\rangle\rangle$ is the
  monoid generated by $S$.
\end{theorem}
\begin{proof}
  First, observe that the fact that $\langle\mul,\shf\rangle$ is of the
  form~\eqref{eqn:mulcirc} is equivalent to stating that
  $\langle\mul,\shf\rangle$ belongs to $S'$. Using this observation
  we prove that $\C[M]$ is indeed an $\mathbf{S'}$-closure operator.
  
  For any $a \in L$, we trivially have $M \models a \Rightarrow a$ and thus
  $a \in \D[M]{a}$, meaning that $a \leq \C[M]{a}$. Thus, $\C[M]$
  satisfies~\eqref{S:ext}.

  In order to prove \eqref{S:mon}, take $a_1,a_2 \in L$ such that
  $a_1 \leq a_2$. Let $b \in \D[M]{a_1}$. By definition, we have
  $M \models a_1 \Rightarrow b$. Now, take any $m \in M$ and
  $\langle\mul,\shf\rangle \in S'$ such that $\mul{a_2} \leq m$.  Using
  $a_1 \leq a_2$ and~\eqref{eqn:mon_mul}, we get
  $\mul{a_1} \leq \mul{a_2} \leq m$ and so $\mul{b} \leq m$.  That is,
  $M \models a_2 \Rightarrow b$, i.e., $b \in \D[M]{a_2}$. Therefore,
  $\D[M]{a_1} \subseteq \D[M]{a_2}$ and thus $\C[M]{a_1} \leq \C[M]{a_2}$,
  proving~\eqref{S:mon}.

  It remains to check that~\eqref{S:idm} is satisfied. Using~\eqref{eqn:DM}
  and~\eqref{eqn:CM}, it suffices to check that $b \leq \shf{\C[M]{a}}$
  provided that $b \in \D[M]{\shf{\C[M]{a}}}$. Note that the fact that
  $b \in \D[M]{\shf{\C[M]{a}}}$ means
  $M \models \shf{\C[M]{a}} \Rightarrow b$. Using Lemma~\ref{le:Mmul}, the
  last observation gives
  $M \models \mul{\shf{\C[M]{a}}} \Rightarrow \mul{b}$.  Furthermore,
  using~\eqref{eqn:mulshf}, is follows that
  $M \models \C[M]{a} \Rightarrow \mul{b}$, see
  Remark~\ref{rem:Msat}\,(b). Applying $M \models a \Rightarrow \C[M]{a}$,
  which is a consequence of Lemma~\ref{le:supclose}, together with the
  transitivity of~$\leq$, we get $M \models a \Rightarrow \mul{b}$. Hence,
  $\mul{b} \leq \C[M]{a}$ and $b \leq \shf{\C[M]{a}}$ owing
  to~\eqref{eqn:gal}.
\end{proof}

Our last observation in this subsection shows that $\C[M]$ can also be
introduced without any explicit reference to $\models$. The next assertion
shows a way to compute $\C[M]{a}$ for any $a \in L$ using the data (i.e.,
the subset $M$) and the utilized parameterization.

\begin{theorem}\label{th:CMalt}
  Let $S$ be an $L$-parameterization and let
  $\mathbf{S'} = \langle S',\circ,\langle\one,\one\rangle\rangle$ be the
  monoid generated by $S$. Then,
  \begin{align}
    \C[M]{a} &=
    \textstyle\bigwedge\{\shf{m};\, 
    m \in M
    \text{, }
    \langle \mul,\shf\rangle \in S'
    \text{, and }
    a \leq \shf{m}\},
    \label{eqn:CMalt}
  \end{align}
  for any $M \subseteq L$ and any $a \in L$.
\end{theorem}
\begin{proof}
  Using Definition~\ref{def:Msat}, \eqref{eqn:gal}, and the fact that
  infima in complete lattices can be expressed as suprema of lower cones,
  we get
  \begin{align*}
    \C[M]{a} &=
    \textstyle\bigvee\{b \in L;\, M \models a \Rightarrow b\} \\
    &=
    \textstyle\bigvee\{b \in L;\,
    \text{for all }
    m \in M, \langle\mul,\shf\rangle \in S'\!:\,
    \mul{a} \leq m \text{ implies } \mul{b} \leq m\} \\
    &=
    \textstyle\bigvee\{b \in L;\,
    \text{for all }
    m \in M,
    \langle\mul,\shf\rangle \in S',
    a \leq \shf{m}\!:\,
    b \leq \shf{m}\} \\
    &=
    \textstyle\bigwedge\{\shf{m};\, 
    m \in M
    \text{, }
    \langle \mul,\shf\rangle \in S'
    \text{, and }
    a \leq \shf{m}\},
  \end{align*}
  which proves the claim.
\end{proof}

\begin{remark}
  Returning to Remark~\ref{rem:Msat}\,(a), it can be shown
  that~\eqref{eqn:CMalt} generalizes the operator which for any
  $A \subseteq Y$ returns the intent generated by $A$. Indeed, taking into
  account the special case described in Remark~\ref{rem:Msat}\,(a),
  applying~\eqref{eqn:CMalt}, we get
  \begin{align*}
    \C[M]{A} &=
    \textstyle\bigcap\bigl\{\{x\}^{\uparrow_I};\,
    x \in X \text{ and } A \subseteq \{x\}^{\uparrow_I}\bigr\}
    \\
    &=
    A^{\downarrow_I\uparrow_I},
  \end{align*}
  where ${}^{\downarrow_I\uparrow_I}$ is the composition of the usual
  concept-forming operators~\cite{Ga:Tbaca,GaWi:FCA}. Hence,
  Corollary~\ref{co:Mchar} and Theorem~\ref{th:CMalt} yield
  $M \models A \Rightarrow B$ if{}f $B \subseteq A^{\downarrow_I\uparrow_I}$
  which is the standard criterion for $A \Rightarrow B$ being true in the
  formal context $I$, see~\cite[p.~80]{GaWi:FCA} for more details.
\end{remark}

\begin{example}
  Our general definition of if-then formulas true in given $M \subseteq L$
  under an $L$-parameterization $S$ encompasses several approaches which
  appeared earlier. We have already shown in Remark~\ref{rem:Msat} that
  this is true for the classic attribute implications. In addition, 
  the present approach generalizes the notions of formulas being true in
  data with graded attributes and temporal data.

  (a) In the approaches to fuzzy attribute implications parameterized by
  hedges and similarity-based functional dependencies with the same type of
  parameterization~\cite{BeVy:Addg2}, we can consider the same
  parameterization as in Example~\ref{ex:syst}\,(b). In that case,
  $A \Rightarrow B$ is true in $I\!: X \times Y \to L$ to degree $1$ if{}f
  $M \models A \Rightarrow B$ in the sense of Definition~\ref{def:Msat} for
  $M = \{I_x;\, x \in X\}$ where $I_x\!: Y \to L$ such that $I_x(y) = I(x,y)$
  for all $x \in X$ and all $y \in Y$, see~\cite{BeVy:ICFCA} for
  details. Analogous observation can be made for the approach to
  similarity-based functional dependencies in~\cite{BeVy:DASFAA}. Let us add
  that both~\cite{BeVy:ICFCA} and~\cite{BeVy:DASFAA} deal with graded notion
  of truth of fuzzy attribute implications, i.e., in general the papers
  consider degrees (other than $0$ and $1$) to which formulas are true in
  given data. By the formalism presented in this section, we capture only the
  concept of being true to degree~$1$ which may be seen as a limitation,
  however, as it is shown in~\cite[Theorem 22]{Vy:Pasefai}, in the presence
  of general parameterizations, the graded notion of truth of fuzzy attribute
  implications can be expressed solely using the concept of full truth (i.e.,
  truth to degree $1$).

  (b) From the point of view of attribute implications annotated by time
  points that have been introduced in~\cite{TrVy:TAsisaits}, we can consider
  the same corresponding parameterization as in Example~\ref{ex:syst}\,(c).
  According to~\cite{TrVy:TAsisaits}, the input data can be seen as a triadic
  context~\cite{LeWi:Tafca} with conditions being time points
  in~$\mathbb{Z}$. Thus, for $I \subseteq X \times Y \times \mathbb{Z}$ and
  $A,B \subseteq Y \times \mathbb{Z}$ (where both $A$ and $B$ are finite),
  $A \Rightarrow B$ is considered true in $I$, see~\cite{TrVy:TAsisaits},
  whenever for every $i \in \mathbb{Z}$ and $x \in X$, we have that
  \begin{align*}
    &\text{if for any }
    \langle y,z\rangle \in A
    \text{ we have }
    \langle x,y,z+i\rangle \in I
    \text{,}
    \\
    &\text{then for any }
    \langle y,z\rangle \in B
    \text{ we have }
    \langle x,y,z+i\rangle \in I
    \text{.}
  \end{align*}
  As in the previous cases, it can be easily seen that $A \Rightarrow B$ is
  true in $I$ in sense of~\cite{TrVy:TAsisaits} if{}f
  $M \models A \Rightarrow B$ according to Definition~\ref{def:Msat} for
  $M = \{I_x;\, x \in X\}$ where
  $I_x = \{\langle y,z\rangle;\, \langle x,y,z\rangle \in I\}$
  for all $x \in X$.
\end{example}

\section{Properties of $\mathbf{S}$-closure operators}\label{sec:props}
In this section, we present further properties of $S$-closure operators. We
present equivalent ways to define such parameterized operators and present
$S$-closure systems which are related to $S$-closure operators in a way which
is analogous to the relationship of the classic closure systems and
operators. In order to simplify notation,
$\mathbf{L} = \langle L,\leq\rangle$ always stands for a complete lattice and
$\langle S,\circ,\langle\one,\one\rangle\rangle$ always stands for an
$L$-parameterization.

\begin{theorem}\label{th:mon_idm}
  An operator $\C\!: L \to L$ is an $S$-closure operator in $\mathbf{L}$
  if{}f it satisfies \eqref{S:ext} and
  \begin{align}
    &a \leq \shf{\C{b}} \text{ implies } \C{a} \leq \shf{\C{b}},
    \label{S:mon_idm}
  \end{align}
  holds for all $a,b \in L$ and all $\langle \mul,\shf\rangle \in S$.
\end{theorem}
\begin{proof}
  First, we show that each $S$-closure operator
  satisfies~\eqref{S:mon_idm}. Suppose that for $a,b \in L$
  we have $a \leq \shf{\C{b}}$. Using~\eqref{S:mon} and~\eqref{S:idm},
  it readily follows that
  \begin{align*}
    &\C{a} \leq \C{\shf{\C{b}}} \leq \shf{\C{b}},
  \end{align*}
  showing that \eqref{S:mon_idm} holds.

  Conversely, we show that assuming~\eqref{S:ext} and~\eqref{S:mon_idm},
  it follows that the operator $\C$ satisfies~\eqref{S:mon} and~\eqref{S:idm}.
  First, assume that $a \leq b$ and using~\eqref{S:ext} and the fact that
  $\langle\one,\one\rangle \in S$, we get
  \begin{align*}
    &a \leq b \leq \C{b} = \one{\C{b}}.
  \end{align*}
  Therefore, applying~\eqref{S:mon_idm} for $\shf = \one$, it follows that
  \begin{align*}
    &\C{a} \leq \one{\C{b}} = \C{b}
  \end{align*}
  which proves that~\eqref{S:mon} is satisfied. In order to show
  that~\eqref{S:idm} holds, it suffices to consider
  $\shf{\C{a}} \leq \shf{\C{a}}$ and apply~\eqref{S:mon_idm}
  to get~\eqref{S:idm}.
\end{proof}

\begin{theorem}\label{th:mon_alt}
  An operator $\C\!: L \to L$
  is an $S$-closure operator in $\mathbf{L}$ if{}f
  it satisfies \eqref{S:ext}, \eqref{S:ord_idm}, and
  \begin{align}
    a \leq \shf{b} \text{ implies } \C{a} \leq \shf{\C{b}}
    \label{S:mon_alt}
  \end{align}
  holds for all $a,b \in L$ and all $\langle\mul,\shf\rangle \in S$.
\end{theorem}
\begin{proof}
  The fact that each $S$-closure operator satisfies~\eqref{S:mon_alt} follows
  by Theorem~\ref{th:mon_idm} using the fact that~\eqref{S:mon_alt} is a
  weaker condition than~\eqref{S:mon_idm}. Indeed, if $\C$ is an $S$-closure
  operator and $a \leq \shf{b}$, then we get $a \leq \C{\shf{b}}$ owing
  to~\eqref{S:ext} and~\eqref{eqn:mon_shf}. Thus, $\C{a} \leq \shf{\C{b}}$
  is a consequence of~\eqref{S:mon_idm}. Conversely, let $\C$
  satisfy~\eqref{S:ext}, \eqref{S:ord_idm}, and~\eqref{S:mon_alt}. Observe
  that the isotony~\eqref{S:mon} follows directly by~\eqref{S:mon_alt} for
  $\shf = \one$. Thus, it suffices to check that $\C$
  satisfies~\eqref{S:idm}. Take any $a \in L$ and
  $\langle\mul,\shf\rangle \in S$. Using $\shf{\C{a}} \leq \shf{\C{a}}$ and
  \eqref{S:mon_alt}, we get
  \begin{align*}
    \C{\shf{\C{a}}} \leq \shf{\C{\C{a}}}.
  \end{align*}
  Therefore, \eqref{S:idm} is a consequence of the previous inequality,
  \eqref{eqn:mon_shf}, and~\eqref{S:ord_idm}.
\end{proof}

\begin{remark}
  The requirement of~\eqref{S:ord_idm} being satisfied in
  Theorem~\ref{th:mon_alt} cannot be dropped. For instance, if
  $S = \{\langle\one,\one\rangle\}$ and $L = \{0,a,1\}$ such that
  $0 < a < 1$, then we can consider an isotone, extensive, and non-idempotent
  operator $\C\!: L \to L$ defined by $\C{0} = a$ and $\C{a} = \C{1} = 1$.
  Obviously, such an operator satisfies both~\eqref{S:ext}
  and~\eqref{S:mon_alt} but it is not an $S$-closure operator.
\end{remark}

Note that using~\eqref{eqn:gal}, we can express~\eqref{S:mon_idm}
and~\eqref{S:mon_alt} in terms of the lower adjoints in $S$ instead of the
upper adjoints in $S$. The following assertions show further properties
of $S$-closure operators related to lower and upper adjoints.

\begin{theorem}\label{th:mulC_Cmul}
  Let $\C\!: L \to L$ be an $S$-closure operator in $\mathbf{L}$. Then,
  \begin{align}
    \mul{\C{a}} &\leq \C{\mul{a}},
    \label{eqn:mulC_leq_Cmul}
    \\
    \C{\shf{a}} &\leq \shf{\C{a}},
    \label{eqn:Cshf_leq_shfC}
    \\
    \mul{\C{\shf{a}}} &\leq \shf{\C{\mul{a}}},
    \label{eqn:mulCshf_leq_shfCmul}
  \end{align}
  hold for all $a \in L$ and any $\langle\mul,\shf\rangle \in S$.
\end{theorem}
\begin{proof}
  In order to prove~\eqref{eqn:mulC_leq_Cmul}, we apply \eqref{S:ext} to get
  $\mul{a} \leq \C{\mul{a}}$ and thus $a \leq \shf{\C{\mul{a}}}$ owing
  to~\eqref{eqn:gal}. Using~\eqref{S:mon}, we further get
  $\C{a} \leq \C{\shf{\C{\mul{a}}}}$. Now, applying~\eqref{S:idm}, it
  follows that $\C{a} \leq \shf{\C{\mul{a}}}$. Hence,
  $\mul{\C{a}} \leq \C{\mul{a}}$ using~\eqref{eqn:gal}.

  To see that~\eqref{eqn:Cshf_leq_shfC} holds, we start with $a \leq \C{a}$,
  which is an instance of~\eqref{S:ext}, and~apply \eqref{eqn:mon_shf} to get
  $\shf{a} \leq \shf{\C{a}}$. Thus, $\C{\shf{a}} \leq \C{\shf{\C{a}}}$ as a
  consequence of~\eqref{S:mon}. Hence, using~\eqref{S:idm}, we obtain
  $\C{\shf{a}} \leq \shf{\C{a}}$ which proves~\eqref{eqn:Cshf_leq_shfC}.

  Now, \eqref{eqn:mulCshf_leq_shfCmul} is a consequence
  of~\eqref{eqn:mulC_leq_Cmul} and~\eqref{eqn:Cshf_leq_shfC}.
  Indeed, \eqref{eqn:mulC_leq_Cmul} and~\eqref{eqn:gal} yield
  $\C{a} \leq \shf{\C{\mul{a}}}$ and analogously 
  \eqref{eqn:Cshf_leq_shfC} and~\eqref{eqn:gal} give
  $\mul{\C{\shf{a}}} \leq \C{a}$. Then, \eqref{eqn:mulCshf_leq_shfCmul}
  follows using the transitivity of $\leq$.
\end{proof}

Since $S$-closure operators as well as all the upper and lower adjoints
are isotone, the inequalities from the previous assertion can be generalized
to if-then conditions which can be seen as generalized isotony conditions.

\begin{corollary}
  Let $\C\!: L \to L$ be an $S$-closure operator in $\mathbf{L}$. Then,
  \begin{align}
    a \leq b &\text{ implies } \mul{\C{a}} \leq \C{\mul{b}},
    \label{eqn:im_mulC_leq_Cmul}
    \\
    a \leq b &\text{ implies } \C{\shf{a}} \leq \shf{\C{b}},
    \label{eqn:im_Cshf_leq_shfC}
    \\
    a \leq b &\text{ implies } \mul{\C{\shf{a}}} \leq \shf{\C{\mul{b}}},
    \label{eqn:im_mulCshf_leq_shfCmul}
  \end{align}
  hold for all $a,b \in L$ and any $\langle\mul,\shf\rangle \in S$.
  \qed
\end{corollary}

\begin{remark}
  The converse inequalities to those in~\eqref{eqn:mulC_leq_Cmul}
  and~\eqref{eqn:Cshf_leq_shfC} do not hold in general. In case
  of~\eqref{eqn:mulC_leq_Cmul}, we can take $L = \{0,1\}$ with $0 < 1$ and
  $\langle\mul,\shf\rangle$ such that $\mul{0} = \mul{1} = 0$ and
  $\shf{0} = \shf{1} = 1$. In this setting, $\C: L \to L$ such that
  $\C{0} = \C{1} = 1$ is an $S$-closure operator for
  $S = \{\langle\mul,\shf\rangle,\langle\one,\one\rangle\}$. In addition, we
  have $\mul{\C{1}} = 0$ and $\C{\mul{1}} = 1$, i.e.,
  $\C{\mul{a}} \nleq \mul{\C{a}}$ for $a = 1$.

  In case of~\eqref{eqn:mulC_leq_Cmul}, consider $L = \{0,a,1\}$ such that
  $0 < a < 1$. Furthermore, a pair $\langle\mul,\shf\rangle$ of maps defined
  by $\mul{0} = \shf{0} = 0$, $\mul{a} = \mul{1} = a$, and
  $\shf{a} = \shf{1} = 1$ is an isotone Galois connection in $\mathbf{L}$.
  In addition, $\C$ defined by $\C{0} = \C{a} = a$ and $\C{1} = 1$ is an
  $S$-closure operator for
  $S = \{\langle\mul,\shf\rangle,\langle\one,\one\rangle\}$ such that
  $\shf{\C{0}} = 1 \nleq a = \C{\shf{0}}$.
\end{remark}

\begin{theorem}\label{th:S_alt_iso}
  Let $\C\!: L \to L$ be a closure operator in $\mathbf{L}$.
  Then, the following conditions are equivalent:
  \begin{enumerate}\parskip=0pt%
  \item[\itm{1}]
    $\C$ is an $S$-closure operator in $\mathbf{L}$;
  \item[\itm{2}]
    $\C$ satisfies \eqref{eqn:mulC_leq_Cmul};
  \item[\itm{3}]
    $\C$ satisfies \eqref{eqn:Cshf_leq_shfC}.
  \end{enumerate}
\end{theorem}
\begin{proof}
  Using Theorem~\ref{th:mulC_Cmul}, it follows that \itm{1} implies both
  \itm{2} and \itm{3}. Therefore, it suffices to show that either of \itm{2}
  and \itm{3} implies \eqref{S:idm}.

  Suppose that \itm{2} holds. First,
  $\mul{\C{\shf{\C{a}}}} \leq \C{\mul{\shf{\C{a}}}}$ holds true because it is
  a particular case of \eqref{eqn:mulC_leq_Cmul}. Furthermore, using
  \eqref{eqn:mulshf} and~\eqref{S:mon}, it follows that
  $\mul{\shf{\C{a}}} \leq \C{a}$ and $\C{\mul{\shf{\C{a}}}} \leq \C{\C{a}}$.
  Applying~\eqref{S:ord_idm} to the last inequality, we get
  $\C{\mul{\shf{\C{a}}}} \leq \C{a}$. Hence, the transitivity of $\leq$
  together with~\eqref{eqn:gal} give $\mul{\C{\shf{\C{a}}}} \leq \C{a}$ and
  $\C{\shf{\C{a}}} \leq \shf{\C{a}}$ which proves \eqref{S:idm}.

  In case of~\itm{3}, the argument is straightforward: We have
  $\C{\shf{\C{a}}} \leq \shf{\C{\C{a}}}$ as a particular case of
  \eqref{eqn:Cshf_leq_shfC}. Thus, using \eqref{S:ord_idm} together with
  \eqref{eqn:mon_shf}, we get $\C{\shf{\C{a}}} \leq \shf{\C{a}}$,
  i.e., \eqref{S:idm} holds.
\end{proof}

Let us recall that the $\mathbf{L}^*$-closure operators and
$\mathbf{L}_K$-closure operators, which are important examples of $S$-closure
operators on complete residuated lattices of fuzzy sets, were defined as
closure operators with stronger isotony conditions. In contrast, $S$-closure
operators have been introduced as operators with the classic isotony
condition and with stronger form of the idempotency. Based on our
observations, $S$-closure operators can also be seen as closure operators
with stronger isotony conditions. Indeed, it is easy to see that
\eqref{eqn:im_mulC_leq_Cmul} and \eqref{eqn:im_Cshf_leq_shfC} imply
\eqref{eqn:mulC_leq_Cmul} and \eqref{eqn:Cshf_leq_shfC}, respectively. In
addition, \eqref{eqn:im_mulC_leq_Cmul} and \eqref{eqn:im_Cshf_leq_shfC} imply
the ordinary isotony condition for $\mul = \one$ and $\shf = \one$,
respectively. Therefore, using Theorem~\ref{th:S_alt_iso},
we come to the following conclusion:

\begin{corollary}
  Let $\C\!: L \to L$ be an operator in $\mathbf{L}$ which
  satisfies~\eqref{S:ext} and~\eqref{S:ord_idm}.  
  Then, the following conditions are equivalent:
  \begin{enumerate}\parskip=0pt%
  \item[\itm{1}]
    $\C$ is an $S$-closure operator in $\mathbf{L}$;
  \item[\itm{2}]
    $\C$ satisfies \eqref{eqn:im_mulC_leq_Cmul};
  \item[\itm{3}]
    $\C$ satisfies \eqref{eqn:im_Cshf_leq_shfC}.
    \qed
  \end{enumerate}
\end{corollary}

We now turn our attention to closure systems related to $S$-closure
operators. The systems can be seen as systems closed under applications of
upper adjoints of all isotone Galois connections in $S$.

\begin{definition}\label{def:S-clos}
  Let $\mathbf{L} = \langle L,\leq\rangle$ be a complete lattice and let $S$
  be an $L$-parameterization. A~system $\mathcal{S} \subseteq L$ is called an
  $S$-closure system in $\langle L, \leq\rangle$ whenever it is a closure
  system in $\langle L,\leq\rangle$ such that $\shf{a} \in \mathcal{S}$ for
  any $\langle\mul,\shf\rangle \in S$ and any $a \in L$. In addition,
  if $\mathbf{S} = \langle S,\circ,\langle\one,\one\rangle\rangle$ is
  an $\mathbf{L}$-parameterization, then $\mathcal{S}$ is called
  an $\mathbf{S}$-closure system.
\end{definition}

The requirements of being closed under arbitrary infima and being closed
under applications of upper adjoints can be replaced by a single condition as
it is shown in the following assertion.

\begin{theorem}
  A system $\mathcal{S} \subseteq L$ is an $S$-closure system in
  $\mathbf{L}$ if{}f
  \begin{align}
    \textstyle\bigwedge\{\shf{b};\,
    a \leq \shf{b}
    \text{, }
    b \in \mathcal{S}
    \text{, and }
    \langle\mul,\shf\rangle \in S\} \in \mathcal{S}
    \label{S:syst}
  \end{align}
  for any $a \in L$.
\end{theorem}
\begin{proof}
  In order to show the only-if part, observe that for each $b \in \mathcal{S}$,
  we have that $\shf{b} \in \mathcal{S}$, i.e., \eqref{S:syst} is a consequence
  of the fact that $\mathcal{S}$ is closed under $\bigwedge$ and applications of
  $\shf$ of any $\langle\mul,\shf\rangle \in S$.

  Conversely, let~$\mathcal{S} \subseteq L$ satisfy \eqref{S:syst} for any
  $a \in L$. Take any $c \in \mathcal{S}$ and $\langle\mul,\shf\rangle \in S$.
  Putting $a = \shf{c}$, \eqref{S:syst} yields
  \begin{align*}
    \shf{c}
    &\leq
    \textstyle\bigwedge\{\shf{b};\,
    \shf{c} \leq \shf{b}
    \text{, }
    b \in \mathcal{S}
    \text{, and }
    \langle\mul,\shf\rangle \in S\} \in \mathcal{S},
  \end{align*}
  because $\shf{c} \leq \shf{b}$ for all $b \in \mathcal{S}$ and
  $\langle\mul,\shf\rangle \in S$ and thus the infimum in the previous
  inequality is an infimum of elements which are greater than or equal
  to $\shf{c}$. Therefore, in order to prove that $\shf{c} \in \mathcal{S}$,
  it suffices to check the converse inequality but this is easy to see since
  $c \in \mathcal{S}$. Indeed, the fact that $c \in \mathcal{S}$ yields
  \begin{align*}
    \shf{c} &\in
    \{\shf{c};\,
    \langle\mul,\shf\rangle \in S\}
    \\
    &=
    \{\shf{c};\,
    \shf{c} \leq \shf{c}
    \text{ and }
    \langle\mul,\shf\rangle \in S\}
    \\
    &\subseteq
    \{\shf{b};\,
    \shf{c} \leq \shf{b}
    \text{, }
    b \in \mathcal{S}
    \text{, and }
    \langle\mul,\shf\rangle \in S\}.
  \end{align*}
  As a consequence, we have
  \begin{align*}
    \shf{c}
    &=
    \textstyle\bigwedge\{\shf{b};\,
    \shf{c} \leq \shf{b}
    \text{, }
    b \in \mathcal{S}
    \text{, and }
    \langle\mul,\shf\rangle \in S\} \in \mathcal{S},
  \end{align*}
  which shows that $\mathcal{S}$ is closed under applications of all $\shf$'s.
  The fact that $\mathcal{S}$ is closed under arbitrary infima can be proved by
  similar arguments: Take $\mathcal{R} \subseteq \mathcal{S}$ and observe that
  \begin{align*}
    \textstyle\bigwedge\mathcal{R} \leq
    \textstyle\bigwedge\{\shf{b};\,
    \textstyle\bigwedge\mathcal{R} \leq \shf{b}
    \text{, }
    b \in \mathcal{S}
    \text{, and }
    \langle\mul,\shf\rangle \in S\} \in \mathcal{S}.
  \end{align*}
  Moreover, for any $c \in \mathcal{R} \subseteq \mathcal{S}$, we have that
  \begin{align*}
    c &\in \{\one{b};\,
    \textstyle\bigwedge\mathcal{R} \leq \one{b},
    \text{ and }
    b \in \mathcal{S}
    \} \\
    &\subseteq
    \{\shf{b};\,
    \textstyle\bigwedge\mathcal{R} \leq \shf{b}
    \text{, }
    b \in \mathcal{S}
    \text{, and }
    \langle\mul,\shf\rangle \in S\}.
  \end{align*}
  As a consequence, for any $c \in \mathcal{R} \subseteq \mathcal{S}$,
  it follows that
  \begin{align*}
    \textstyle\bigwedge\{\shf{b};\,
    \textstyle\bigwedge\mathcal{R} \leq \shf{b}
    \text{, }
    b \in \mathcal{S}
    \text{, and }
    \langle\mul,\shf\rangle \in S\} \leq c
  \end{align*}
  Therefore, we have
  \begin{align*}
    \textstyle\bigwedge\{\shf{b};\,
    \textstyle\bigwedge\mathcal{R} \leq \shf{b}
    \text{, }
    b \in \mathcal{S}
    \text{, and }
    \langle\mul,\shf\rangle \in S\} \leq
    \textstyle\bigwedge\mathcal{R},
  \end{align*}
  Now, we use~\eqref{S:syst} for $a = \textstyle\bigwedge\mathcal{R}$
  to conclude that $\textstyle\bigwedge\mathcal{R} \in \mathcal{S}$.
\end{proof}

\begin{remark}
  Let us note that $S$-closure operators and $S$-closure systems in
  $\mathbf{L}$ are related in a way which is fully analogous to the
  relationship of the classic closure operators and closure systems. That is,
  given an $S$-closure operator $\C$, the set
  $\mathcal{S}_{\C} = \{a \in L;\, \C{a} = a\} = \{\C{a};\, a \in L\}$ of all
  fixed points of $\C$ is an $S$-closure system; given an $S$-closure system
  $\mathcal{S}$, an operator $\C[\mathcal{S}]\!: L \to L$ defined by
  $\C[\mathcal{S}](a) = \bigwedge \{b \in \mathcal{S};\, a \leq b\}$ for any
  $a \in L$ is an $S$-closure operator. Furthermore, we have that
  $\mathcal{S} = \mathcal{S}_{\C[\mathcal{S}]}$ and
  $\C = \C[\mathcal{S}_{\C}]$. These claims are routine to check and use
  similar arguments as in the classic setting.
\end{remark}

We conclude this section with notes on a partial order which can be
introduced on the set of all isotone Galois connections in $\mathbf{L}$ and
its relationship to $\mathbf{L}$-parameterizations. For any isotone Galois
connections $\langle\mul[1],\shf[1]\rangle$ and
$\langle\mul[2],\shf[2]\rangle$ in $\mathbf{L}$, we put
$\langle\mul[1],\shf[1]\rangle \leqslant \langle\mul[2],\shf[2]\rangle$
whenever
\begin{align}
  \mul[1]{a} \leq \mul[2]{a}
  \label{eqn:leq}
\end{align}
for all $a \in L$. Obviously, $\leqslant$ defined by~\eqref{eqn:leq} is a
partial order on the set of all isotone Galois connections in
$\mathbf{L}$. Indeed, its reflexivity and transitivity follow directly from
the properties of the complete lattice order $\leq$ in $\mathbf{L}$. In
addition, the antisymmetry of $\leqslant$ follows by applying the
antisymmetry of $\leq$ together with the fact that each isotone Galois
connection is uniquely given by its lower adjoint. It is easy to see that
$\leqslant$ can be defined in an equivalent way in terms of the upper
adjoints in $S$ instead of the lower ones:

\begin{lemma}\label{le:leq_alt}
  For $\leqslant$ defined as in~\eqref{eqn:leq}, we have
  $\langle\mul[1],\shf[1]\rangle \leqslant \langle\mul[2],\shf[2]\rangle$
  if{}f
  \begin{align}
    \shf[2]{a} \leq \shf[1]{a}
    \label{eqn:alt_leq}
  \end{align}
  for all $a \in L$.
\end{lemma}
\begin{proof}
  Let $\mul[1]{a} \leq \mul[2]{a}$ for all $a \in L$. In particular,
  for $a = \shf[2]{b}$ where $b \in L$, using~\eqref{eqn:mon_mul}
  and~\eqref{eqn:mulshf}, it follows that
  \begin{align*}
    \mul[1]{\shf[2]{b}} \leq \mul[2]{\shf[2]{b}} \leq b.
  \end{align*}
  Therefore, by adjointness, we get $\shf[2]{b} \leq \shf[1]{b}$. Since
  $b \in L$ was taken arbitrarily, this proves that~\eqref{eqn:alt_leq} holds.
  Analogously, we can check that assuming~\eqref{eqn:alt_leq}, we get
  that~\eqref{eqn:leq} holds.
\end{proof}

The following assertion shows that $\leqslant$ defined by~\eqref{eqn:leq} is
compatible with composition of isotone Galois connections.

\begin{lemma}\label{le:pomonoid}
  Let $\leqslant$ be defined as in~\eqref{eqn:leq}. If
  $\langle\mul[1],\shf[1]\rangle \leqslant \langle\mul[2],\shf[2]\rangle$ and
  $\langle\mul[3],\shf[3]\rangle \leqslant \langle\mul[4],\shf[4]\rangle$,
  then
  \begin{align*}
    \langle\mul[1],\shf[1]\rangle \circ \langle\mul[3],\shf[3]\rangle
    \leqslant
    \langle\mul[2],\shf[2]\rangle \circ \langle\mul[4],\shf[4]\rangle.
  \end{align*}
\end{lemma}
\begin{proof}
  Take any $a \in L$. Since $\langle\mul[3],\shf[3]\rangle \leqslant
  \langle\mul[4],\shf[4]\rangle$, we have $\mul[3]{a} \leq
  \mul[4]{a}$. Furthermore, $\langle\mul[1],\shf[1]\rangle \leqslant
  \langle\mul[2],\shf[2]\rangle$ applied together with~\eqref{eqn:mon_mul}
  yields
  \begin{align*}
    \mul[1]{\mul[3]{a}} \leq
    \mul[2]{\mul[3]{a}} \leq
    \mul[2]{\mul[4]{a}},
  \end{align*}
  The rest follows from the fact that
  $a \in L$ was chosen arbitrarily.
\end{proof}

As a consequence of the previous lemma, in case of
$\mathbf{L}$-parameterizations, i.e., $L$-parameterizations which are closed
under $\circ$, we in fact deal with partially-ordered monoids, i.e.,
partially-ordered algebras whose operations are compatible with the
underlying partial order:

\begin{corollary}\label{cor:pomonoid}
  Let $\mathbf{S} = \langle S,\circ,\langle\one,\one\rangle\rangle$ be an
  $\mathbf{L}$-parameterization and denote by $\leqslant$ the relation on $S$
  defined as in~\eqref{eqn:leq}. Then, 
  $\mathbf{S} = \langle S,\leqslant,\circ,\langle\one,\one\rangle\rangle$
  is a partially-ordered monoid.
  \qed
\end{corollary}

\begin{remark}\label{rem:extremal}
  Depending on particular families of $\mathbf{L}$-parameterizations, we can
  get even stronger properties of
  $\mathbf{S} = \langle S,\leqslant,\circ,\langle\one,\one\rangle\rangle$.
  In this remark, we comment on properties related to the existence of
  extremal elements with respect to $\leqslant$.

  (a) Put $\mul[\bot](a) = 0$ and $\shf[\bot](a) = 1$ for all $a \in L$.
  Clearly, $\langle\mul[\bot],\shf[\bot]\rangle$ is an isotone Galois
  connection in $\mathbf{L}$. If $\langle\mul[\bot],\shf[\bot]\rangle \in S$,
  then it is the least element of $S$ with respect to $\leqslant$. Thus, each
  $L$-parameterization can be extended to an
  $L$-parameterization with a least element.

  (b)
  Consider $\mul[\top]$ and $\shf[\top]$ defined by
  \begin{align}
    \mul[\top]{a} &=
      \begin{cases}
        0, &\text{if } a = 0, \\
        1, &\text{otherwise,}
      \end{cases}
    &\shf[\top]{a} &=
      \begin{cases}
        1, &\text{if } a = 1, \\
        0, &\text{otherwise.}
      \end{cases}
  \end{align}
  for all $a \in L$. Recall that in the context of fuzzy sets, $\shf[\top]$
  corresponds to the globalization~\cite{TaTi:Gist}, cf.~\eqref{eqn:glob} and
  $\mul[\top]$ is a dual hedge~\cite{Vy:Tdhbl} which was also investigated on
  linear G\"odel chains in~\cite{Baaz}.  By moment's reflection, we get that
  $\langle\mul[\top],\shf[\top]\rangle$ is an isotone Galois connection in
  $\mathbf{L}$. In addition, for any $a \in L$ and any isotone Galois
  connection $\langle\mul,\shf\rangle$ in~$\mathbf{L}$, we have
  $\mul{a} \leq \mul[\top]{a}$ (observe that \eqref{eqn:gal} and
  $0 \leq \shf{0}$ give $\mul{0} = 0$). Thus, if
  $\langle\mul[\top],\shf[\top]\rangle \in S$, then it is the greatest
  element of $S$ with respect to $\leqslant$. Recall that
  $\langle\one,\one\rangle$ is neutral with respect to $\circ$ but as we have
  just shown, it may not be the greatest element of $S$ with respect to
  $\leqslant$. As a consequence, partially-ordered monoids of
  $\mathbf{L}$-parameterizations described in Corollary~\ref{cor:pomonoid}
  are not integral~\cite{GaJiKoOn:RL} in general. Note that
  $\langle\one,\one\rangle$ is the greatest element of $S$ if{}f
  $\mul{a} \leq a$ for all $a \in L$ and all $\langle\mul,\shf\rangle \in S$,
  i.e., if{}f each lower adjoint in $S$ is subdiagonal and, equivalently,
  each upper adjoint in $S$ is superdiagonal.
\end{remark}

In case of $\mathbf{S}$-closure operators which are $\mathbf{L}^*$-closure
operators, the underlying $\mathbf{L}$-parameterizations are in fact complete
residuated lattices as we show in the next assertion. In other words, if we
consider any $\mathbf{L}^*$-closure operator $\C: L^Y \to L^Y$ (a~fuzzy
closure operator with the monotony condition parameterized by~${}^*$) and
view it as a special case of an $\mathbf{S}$-closure operator with
$\mathbf{S}$ being an $\mathbf{L}$-parameterization as in
Corollary~\ref{cor:L*_S}, then $\mathbf{S}$ is not only a partially-ordered
monoid but in addition $\leqslant$ is a complete lattice order on $S$ and the
composition $\circ$ of isotone Galois connections restricted to $S$ has its
residuum satisfying the adjointness property~\eqref{eqn:adj}:

\begin{theorem}\label{th:L*_S_reslat}
  Let $\mathbf{L}$ be a complete residuated lattice,
  ${}^*$ be an idempotent truth-stressing hedge, and let 
  $S = \{\langle \mul[a^* \otimes], \shf[a^* \rightarrow]\rangle;\, a \in L\}$.
  Then,
  \begin{align}
    \mathbf{S} = \langle S,\leqslant,\circ,\rightsquigarrow,
    \langle\mul[\bot],\shf[\bot]\rangle,\langle\one,\one\rangle\rangle,
    \label{eqn:L*_S_reslat}
  \end{align}
  where $\rightsquigarrow$ is a binary operation on $S$ defined by
  \begin{align}
    \langle\mul[a^*\otimes],\shf[a^*\rightarrow]\rangle
    \!\rightsquigarrow\!
    \langle\mul[b^*\otimes],\shf[b^*\rightarrow]\rangle =
    \textstyle\bigvee\{
    \langle\mul[c^*\otimes],\shf[c^*\rightarrow]\rangle;\,
    a^* \otimes c^* \leq b^* \text{, } c \in L\},
    \label{eqn:L*_S_residuum}
  \end{align}
  for all $a,b \in L$,
  is a complete (commutative integral) residuated lattice.
\end{theorem}
\begin{proof}
  First, observe that $\circ$ given by~\eqref{eqn:compose} and restricted to
  $S$ is obviously commutative which follows directly using the commutativity
  of~$\otimes$ in $\mathbf{L}$. Furthermore,
  $\langle\mul[\bot],\shf[\bot]\rangle = \langle \mul[0^* \otimes], \shf[0^*
  \rightarrow]\rangle$ and
  $\langle\one,\one\rangle = \langle \mul[1^* \otimes], \shf[1^*
  \rightarrow]\rangle$ are the least and the greatest elements of $S$,
  respectively, cf. Remark~\ref{rem:extremal}\,(b). Therefore, $\mathbf{S}$
  is a bounded commutative integral partially-ordered monoid. Furthermore, it
  is also easy to see that for any $a,b \in L$, we have
  \begin{align}
    \langle\mul[a^*\otimes],\shf[a^*\rightarrow]\rangle
    \leqslant
    \langle\mul[b^*\otimes],\shf[b^*\rightarrow]\rangle
    \text{ if{}f }
    a^* \leq b^*.
    \label{eqn:leq_iff}
  \end{align}
  Moreover, $S$ is closed under arbitrary suprema---this follows from the fact
  that a truth-stressing hedge on $\mathbf{L}$
  satisfying~\eqref{ts:1}--\eqref{ts:idm} is an interior operator,
  i.e., $\bigvee\{a^*_i;\, i \in I\}$ is
  a fixed point of ${}^*$. Hence,
  \begin{align*}
    \textstyle\bigvee\bigl\{
    \bigl\langle\mul[a^*_i\otimes],\shf[a^*_i\rightarrow]\bigr\rangle;\,
    i \in I\bigr\} =
    \bigl\langle\mul[\bigvee\{a^*_i;\, i \in I\}\otimes],
    \shf[\bigvee\{a^*_i;\, i \in I\}\rightarrow]\bigr\rangle \in S.
  \end{align*}
  In order to show that $\circ$ and $\rightsquigarrow$ defined
  by~\eqref{eqn:L*_S_residuum} satisfy the adjointness property, we first
  show that $\circ$ is distributive over $\bigvee$ which is a condition
  equivalent to stating that $\circ$ has a residuum satisfying the
  adjointness property. For any $a \in L$ and any $b_i \in L$ ($i \in I$),
  our previous observations together with~\eqref{eqn:compose} and the fact
  that $\otimes$ is distributive over $\bigvee$ in $\mathbf{L}$ yield
  \begin{align*}
    &\langle\mul[a^*\otimes],\shf[a^*\rightarrow]\rangle
    \circ
    \textstyle\bigvee\bigl\{
    \bigl\langle\mul[b^*_i\otimes],\shf[b^*_i\rightarrow]\bigr\rangle;\,
    i \in I\bigr\} = \\
    &\langle\mul[a^*\otimes],\shf[a^*\rightarrow]\rangle
    \circ
    \bigl\langle\mul[\bigvee\{b^*_i;\, i \in I\}\otimes],
    \shf[\bigvee\{b^*_i;\, i \in I\}\rightarrow]\bigr\rangle = \\
    &
    \bigl\langle\mul[a^*\otimes\bigvee\{b^*_i;\, i \in I\}\otimes],
    \shf[(a^*\otimes\bigvee\{b^*_i;\, i \in I\})\rightarrow]\bigr\rangle = \\
    &
    \bigl\langle\mul[\bigvee\{a^*\otimes b^*_i;\, i \in I\}\otimes],
    \shf[\bigvee\{a^*\otimes b^*_i;\, i \in I\}\rightarrow]\bigr\rangle = \\
    &\textstyle\bigvee\bigl\{
    \langle\mul[a^*\otimes b^*_i \otimes],
      \shf[(a^*\otimes b^*_i)\rightarrow]\rangle;\,
    i \in I\bigr\} = \\
    &\textstyle\bigvee\bigl\{
    \langle\mul[a^*\otimes],\shf[a^*\rightarrow]\rangle
    \circ
    \bigl\langle\mul[b^*_i\otimes],\shf[b^*_i\rightarrow]\bigr\rangle;\,
    i \in I\bigr\}.
  \end{align*}
  Hence, there is uniquely given $\twoheadrightarrow$ such that $\circ$ and
  $\twoheadrightarrow$ satisfy~\eqref{eqn:adj}.
  Using standard properties of complete residuated lattices~\cite{GaJiKoOn:RL},
  we have
  \begin{align*}
    &\langle\mul[a^*\otimes],\shf[a^*\rightarrow]\rangle
      \!\twoheadrightarrow\!
    \langle\mul[b^*\otimes],\shf[b^*\rightarrow]\rangle =
    \\
    &\textstyle\bigvee\{
    \langle\mul[c^*\otimes],\shf[c^*\rightarrow]\rangle \in S;\,
    \langle\mul[a^*\otimes],\shf[a^*\rightarrow]\rangle
    \circ
    \langle\mul[c^*\otimes],\shf[c^*\rightarrow]\rangle
    \leqslant
    \langle\mul[b^*\otimes],\shf[b^*\rightarrow]\rangle\}.
  \end{align*}
  Therefore, it remains to show that exactly $\rightsquigarrow$
  defined by~\eqref{eqn:L*_S_residuum} is the
  residuum of $\circ$ in $\mathbf{L}$, i.e., $\rightsquigarrow$ coincides
  with $\twoheadrightarrow$. Using~\eqref{eqn:leq_iff} together with the
  definition of $\leqslant$ and the latter observation, we get
  \begin{align*}
    &\langle\mul[a^*\otimes],\shf[a^*\rightarrow]\rangle
    \!\twoheadrightarrow\!
    \langle\mul[b^*\otimes],\shf[b^*\rightarrow]\rangle =
    \\
    &\textstyle\bigvee\{
    \langle\mul[c^*\otimes],\shf[c^*\rightarrow]\rangle \in S;\,
    \langle\mul[a^*\otimes],\shf[a^*\rightarrow]\rangle
    \circ
    \langle\mul[c^*\otimes],\shf[c^*\rightarrow]\rangle
    \leqslant
    \langle\mul[b^*\otimes],\shf[b^*\rightarrow]\rangle\} =
    \\
    &\textstyle\bigvee\{
    \langle\mul[c^*\otimes],\shf[c^*\rightarrow]\rangle \in S;\,
    \langle\mul[a^*\otimes c^* \otimes],
    \shf[(a^* \otimes c^*)\rightarrow]\rangle
    \leqslant
    \langle\mul[b^*\otimes],\shf[b^*\rightarrow]\rangle\} =
    \\
    &\textstyle\bigvee\{
    \langle\mul[c^*\otimes],\shf[c^*\rightarrow]\rangle \in S;\,
    a^* \otimes c^* \leq b^*\} =
    \\
    &\textstyle\bigvee\{
    \langle\mul[c^*\otimes],\shf[c^*\rightarrow]\rangle;\,
    a^* \otimes c^* \leq b^* \text{ and } c \in L\}
  \end{align*}
  which shows that $\rightsquigarrow$ defined by~\eqref{eqn:L*_S_residuum}
  coincides with $\twoheadrightarrow$.
\end{proof}

\section{Conclusion}\label{sec:conclusion}
In this paper, we have extended and developed the theory of closure operators
and corresponding closure systems parameterized by systems of isotone Galois
connections. We have shown that the parameterizations may be viewed from two
basic standpoints: First, as requirements on stronger isotony conditions of
closure operators. Second, as requirements on stronger idempotency conditions
of closure operators. From the point of view of the corresponding closure
systems, the parameterizations represent additional requirements on
relationship between elements in a closure system: The presence of an element
in a parameterized closure system implies the presence of other elements
obtained by applying upper adjoints of the utilized parameterization. In
addition to the investigation of properties of the closure structures, we
have presented two extended examples of areas where such operators naturally
appear. As we have shown, the operators appear as operators of syntactic
consequence in various types of non-classical logics of if-then rules which,
as special cases, include logics of fuzzy and temporal if-then rules. We have
also shown that such operators naturally emerge in the analysis of
dependencies in data and, again, as special cases, several approaches in the
analysis of graded and temporal data can be seen as special cases.

Future research in several directions is promising. First, it may be
interesting to further investigate properties of structures of
$\mathbf{L}$-parameterizations based on ways in which the parameterizations
can be induced. The first step in this direction can be found in
Section~\ref{sec:props}, most notably in Theorem~\ref{th:L*_S_reslat} saying
that for parameterizations induced by hedges the structure is a complete
residuated lattice. Second, a thorough analysis of properties which are
necessary and sufficient in order to establish analogs of results known in
further applications would be desirable. Third, the class of
parameterizations studied in this paper is much more general than those
studied earlier. Examples and applications in other areas than fuzzy and
temporal logics may be anticipated.

\subsubsection*{Acknowledgment}
Supported by grant no. \verb|P202/14-11585S| of the Czech Science Foundation.


\footnotesize
\bibliographystyle{amsplain}
\bibliography{cspsigc}

\end{document}